\documentclass[a4paper,onecolumn,11pt,accepted=2023-09-27]{quantumarticle}
\pdfoutput=1
\usepackage[utf8]{inputenc}
\usepackage[english]{babel}
\usepackage[T1]{fontenc}
\usepackage{pifont}
\newcommand{\cmark}{\textcolor{green}{\ding{51}}}%
\newcommand{\xmark}{\textcolor{red}{\ding{55}}}%
\usepackage{amsmath}
\usepackage{hyperref}
\usepackage{tikz}
\usepackage{lipsum}
\usepackage{color, xcolor, colortbl}
\usepackage{graphicx}
\usepackage{geometry}
\usepackage{amsthm,amsmath,amssymb,amsfonts}
\usepackage{dcolumn}
\usepackage{url}
\usepackage{epstopdf}
\usepackage{enumitem}
\usepackage{algorithm}
\usepackage{algpseudocode}
\usepackage{bm}
\usepackage[caption=false]{subfig}
\usepackage{appendix}
\usepackage{multirow}
\usepackage{braket}
\usepackage[english]{babel}
\usepackage{hyperref}
\usepackage[capitalize]{cleveref}
\usepackage{xpatch}
\usepackage{tikz}
\usepackage{adjustbox}
\usepackage{xspace}
\usepackage{placeins}
\usetikzlibrary{quantikz}
\newcommand{\bvec}[1]{\mathbf{#1}}

\newcommand{\vr}{\bvec{r}}

\newcommand{\vtheta}{\boldsymbol \theta}

\newcommand{\polylog}{\operatorname{poly}\log}

\DeclareMathOperator*{\argmax}{arg\,max}
\DeclareMathOperator*{\argmin}{arg\,min}

\usepackage{bbding}

\newcommand{\mc}[1]{\mathcal{#1}}

\newcommand{\wt}[1]{\widetilde{#1}}

\newcommand{\innerp}[2]{\left\langle #1 \vert #2 \right\rangle}
\newcommand{\abs}[1]{\left\lvert#1\right\rvert}

\newcommand{\rd}{\,\mathrm{d}}
\newcommand{\Or}{\mathcal{O}}

\newcommand{\CC}{\mathbb{C}}
\newcommand{\ZZ}{\mathbb{Z}}

\newtheorem{thm}{\protect\theoremname}
\newtheorem{lem}[thm]{\protect\lemmaname}
\newtheorem{rem}[thm]{\protect\remarkname}
\newtheorem{prop}[thm]{\protect\propositionname}

\providecommand{\definitionname}{Definition}
\providecommand{\assumptionname}{Assumption}
\providecommand{\corollaryname}{Corollary}
\providecommand{\lemmaname}{Lemma}
\providecommand{\propositionname}{Proposition}
\providecommand{\remarkname}{Remark}
\providecommand{\theoremname}{Theorem}
\usepackage{bbm}

\makeatletter
\newenvironment{breakablealgorithm}
  {
   \begin{center}
     \refstepcounter{algorithm}
     \hrule height.8pt depth0pt \kern2pt
     \renewcommand{\caption}[2][\relax]{
       {\raggedright\textbf{\fname@algorithm~\thealgorithm} ##2\par}%
       \ifx\relax##1\relax 
         \addcontentsline{loa}{algorithm}{\protect\numberline{\thealgorithm}##2}%
       \else 
         \addcontentsline{loa}{algorithm}{\protect\numberline{\thealgorithm}##1}%
       \fi
       \kern2pt\hrule\kern2pt
     }
  }{
     \kern2pt\hrule\relax
   \end{center}
  }
\makeatother

\usetikzlibrary{fit}
\tikzset{%
  highlight/.style={rectangle,rounded corners,fill=blue!15,draw,fill opacity=0.3,thick,inner sep=0pt}
}

\graphicspath{{figures/}}
\begin{document}

\title{Simultaneous estimation of multiple eigenvalues with short-depth quantum circuit on early fault-tolerant quantum computers}
\newcommand{\DeptMath}{Department of Mathematics, University of California, Berkeley, CA 94720, USA}
\newcommand{\LBLMath}{Applied Mathematics and Computational Research Division, Lawrence Berkeley National Laboratory, Berkeley, CA 94720, USA}
\newcommand{\BCQM}{Berkeley Center for Quantum Information and Computation, Berkeley, CA 94720, USA}
\newcommand{\CIQC}{Challenge Institute of Quantum Computation, University
of California, Berkeley, CA 94720, USA}

\author{Zhiyan Ding}
\email{zding.m@math.berkeley.edu}
\affiliation{\DeptMath}
\author{Lin Lin}
\email{linlin@math.berkeley.edu}
\affiliation{\DeptMath}
\affiliation{\LBLMath}
\affiliation{\CIQC}
\maketitle

\begin{abstract}
We introduce a multi-modal, multi-level quantum complex exponential least squares (MM-QCELS) method to simultaneously estimate multiple  eigenvalues of a quantum Hamiltonian on early fault-tolerant quantum computers. Our theoretical analysis demonstrates that the algorithm exhibits Heisenberg-limited scaling in terms of circuit depth and total cost. Notably, the proposed quantum circuit utilizes just one ancilla qubit, and with appropriate initial state conditions, it achieves significantly shorter circuit depths compared to circuits based on quantum phase estimation (QPE). Numerical results suggest that compared to QPE, the circuit depth can be reduced by around two orders of magnitude under several settings for estimating ground-state and excited-state energies of certain quantum systems.
\end{abstract}

\section{Introduction}\label{sec:intro}

The estimation of ground-state energies and excited-state energies of a Hamiltonian is a fundamental problem in quantum physics with numerous practical applications, such as in the design of new materials, drug discovery, and optimization problems. While the ground-state is often the state of interest for many quantum systems, excited-state energies also provide crucial information for understanding the electronic and optical properties of materials. Classical computers are limited in their ability to accurately estimate these energies for large-scale systems, and quantum computers have the potential to provide a significant speedup in solving such problems. Therefore, developing efficient and accurate methods for estimating ground-state and excited-state energies on quantum computers has become a major area of research in quantum information science.

When estimating multiple eigenvalues on quantum computers, there are  two different strategies to consider. The first method involves preparing a variety of initial states that approximate different target eigenstates, and then estimating each eigenvalue one by one. The second method is to prepare a \textit{single} initial state that has nontrivial overlap with all  eigenstates of interest and estimate the eigenvalues simultaneously. The effectiveness of each approach depends on the assumptions and qualities of the initial state. This paper concerns the second approach. 
Given a quantum Hamiltonian $H\in\mathbb{C}^{M\times M}$, we assume that we can prepare an initial quantum state $\ket{\psi}$ with $K$ \textit{dominant} modes. Specifically, let $\left\{(\lambda_m,\ket{\psi_m})\right\}^M_{m=1}$ denote the eigenvalue and eigenvector pairs of $H$, and define $p_m=\left|\innerp{\psi_m}{\psi}\right|^2$ as the overlap between the initial state and the $m$-th eigenvector.
We assume that we can choose a set $\mathcal{D}\subset \{1,2,\cdots,M\}$ satisfying $\abs{\mc{D}}=K$, and  $p_m=\Omega(R^{(K)})$ for any $m\in\mathcal{D}$. Here $R^{(K)}=\sum_{m'\in\mathcal{D}^c}p_{m'}$ is called the residual overlap, and $\mathcal{D}^c=\{1,2,\cdots,M\}\setminus\mathcal{D}$. The eigenvalues $\{\lambda_m\}_{m\in \mc{D}}$ are called the dominant eigenvalues of $H$ with respect to the initial state $\ket{\psi}$ (or dominant eigenvalues for short), and the associated eigenvectors are called the dominant eigenvectors (or dominant modes). We assume $\{\lambda_m\}_{m\in\mathcal{D}}\subset[-\pi,\pi]$ for simplicity. Using an oracle access to the Hamiltonian evolution operators $e^{-i t H}$ for any $t\in\mathbb{R}$, we introduce an efficient quantum algorithm to estimate these dominant eigenvalues. 
We quantify the efficiency of a quantum algorithm by means of the maximal runtime denoted by $T_{\max}$,
and the total runtime $T_{\mathrm{total}}$. Assuming an algorithm needs to run a set of  Hamiltonian evolution operators $\{\exp(-it_nH)\}^N_{n=1}$, then the maximal runtime is $T_{\max}=\max_{1\leq n\leq N}|t_n|$ and the total runtime is $T_{\mathrm{total}}=\sum^N_{n=1}|t_n|$. Here, $T_{\max}$ and $T_{\mathrm{total}}$ approximately measure the circuit depth and the total cost of the algorithm, respectively, in a way that is oblivious to the details in implementing the Hamiltonian evolution operator $e^{-i t H}$.

Our algorithm satisfies the following properties:
\begin{itemize}
    \item[(1)] Allow a nonzero residual overlap: $R^{(K)}=\sum_{m'\in\mathcal{D}^c}p_{m'}>0$.
    \item[(2)] Maintain  Heisenberg-limited scaling~\cite{Gio2011},\cite{Zwi2010},\cite{Zwi2012}:  To estimate all dominant eigenvalues to precision $\epsilon$ with probability $1-\eta$, the total cost is $\Or(\epsilon^{-1}\polylog(\epsilon^{-1}\eta^{-1}))$;
    \item[(3)] Use one ancilla qubit.
    \item[(4)] Reduce the circuit depth: the maximal runtime can be (much) lower than $\pi/\epsilon$, especially when the ratio of residual overlap to the minimum overlap in the dominant set $R^{(K)}/(\min_{m\in\mathcal{D}}p_m)$ approaches zero.
    \item[(5)] Use the information of the spectral gap to further reduce maximal runtime: in the presence of a spectral gap $\Delta$, when $\epsilon\ll \Delta$, the maximal runtime can be independent of $\epsilon$ as $\widetilde{\mathcal{O}}(\Delta^{-1})$, and in this case, the total runtime is $\widetilde{\mathcal{O}}(\epsilon^{-2})$
(see detail in Remark \ref{re:add}).
\end{itemize}

Our algorithm is designed to perform only classical postprocessing on the quantum data that is generated by the Hadamard test circuit, which uses one ancilla qubit. It is particularly well-suited for early fault-tolerant quantum devices that may be limited in the number of ancilla qubits and  maximal coherence times.

The structure of this paper is organized as follows: In \cref{sec:main_idea}, we present the main idea of our method and provide a brief summary of related works. Next, in \cref{sec:main_method}, we describe the main algorithm and provide the complexity results. The proof is included in the Appendix. Finally, we showcase the efficiency of our algorithm through numerical examples in \cref{sec:ns}.

\section{Main idea and related work}\label{sec:main_idea}

In order to illustrate the main idea of the algorithm, we assume that we can use a quantum circuit (see \cref{sec:mqc}) to accurately estimate
\begin{equation}\label{eqn:signal}
\left\langle\psi\right|\exp(-it H)\ket{\psi}=\sum^{M}_{m=1}p_m\exp(-i \lambda_m t)\,
\end{equation}
for any $t\in\mathbb{R}$, where $t$ is drawn from a probability distribution with a symmetric density $a(t)$, i.e., $a(t)\geq 0$, $a(t)=a(-t)$, and $\|a(t)\|_{L^1}=\int_{-\infty}^{\infty} a(t) \mathrm{d} t =1$. In our implementation, to estimate \eqref{eqn:signal}, we first use a classical computer to randomly sample $t$ from distribution $a(t)$. Then, we execute the Hadamard test circuit (\cref{fig:qc}) on a quantum computer multiple times and average the output of measurement to estimate \eqref{eqn:signal}. We emphasize that the random selection of $t$ plays an important role in improving the efficiency of our algorithm (see \cref{sec:related}).

Define the Fourier transform of the probability density \begin{equation}
F(x)=\int^\infty_{-\infty}a(t)\exp(ixt)\rd t\,,
\end{equation}
which is a real-valued even function. As detailed in the following explanation, our approach requires carefully selecting an appropriate distribution $a(t)$ to ensure that the resulting filter function $F(x)$ reaches its maximum at $x=0$ and decay rapidly as $|x|$ increases. There are several choices of $a(t)$ such that $F(x)$ concentrates around $x=0$. In this paper, we choose $a(t)$ as the density function of a truncated Gaussian function:\begin{equation}\label{eqn:a_T}
    a(t)=\frac{1}{A_\gamma \sqrt{2\pi}T}\exp\left(-\frac{t^2}{2T^2}\right)\textbf{1}_{[-\gamma T,\gamma T]}(t), \quad \mbox{for}\quad T,\gamma>0\,,
\end{equation} 
where the normalization constant $A_\gamma=\int^{\gamma}_{-\gamma}\frac{1}{\sqrt{2\pi}}\exp\left(-\frac{t^2}{2}\right)dt$. The selection of $a(t)$ is guided by the observation that the Fourier transform of a Gaussian function remains a Gaussian function. More specifically, when $\gamma=\infty$, we have $F(x)=\exp\left(-\frac{T^2x^2}{2}\right)$, which attains its maximal value at $x=0$ and exponentially decays to zero with respect to $T|x|$. In addition, the use of the truncated Gaussian ensures that the maximal runtime never exceeds $\gamma T$. We may also use the notation $a_T(t)=a(t), F_T(x)=F(x)$ to emphasize the dependence on $T$.

From this we can define an \textit{ideal} loss function as 
\begin{equation}\label{eqn:loss_multi_modal_perfect}
\mathcal{L}_{K}\left(\{r_k\}^{K}_{k=1},\{\theta_k\}^{K}_{k=1}\right)=\int^\infty_{-\infty}a(t)\left|\sum^{M}_{m=1}p_m\exp(-i \lambda_m t)-\sum^{K}_{k=1}r_k\exp(-i\theta_k t)\right|^2\rd t\,.
\end{equation}
Then we can we estimate the dominant eigenvalues $\{\lambda_m\}_{m\in\mathcal{D}}$ by solving the optimization problem
\begin{equation}\label{eqn:op_perfect}
\left(\{r^*_k\}^{K}_{k=1},\{\theta^*_k\}^{K}_{k=1}\right)=\argmin_{r_k\in\CC,\theta_k\in\mathbb{R}}\mathcal{L}_{K}\left(\{r_k\}^{K}_{k=1},\{\theta_k\}^{K}_{k=1}\right)\,.
\end{equation}

Intuitively, if the residual overlap $R^{(K)}=\sum_{m'\in\mathcal{D}^c}p_{m'}$ is small enough, to minimize the loss function, $\{(r^*_k,\theta^*_k)\}^{K}_{k=1}$ should be close to $\{(p_m,\lambda_m)\}_{m\in\mathcal{D}}$ so that the optimizer $\sum^{K}_{k=1}r^*_k\exp(-i\theta^*_k t)$ can eliminate the dominant term $\sum_{m\in\mathcal{D}}p_m\exp(-i\lambda_m t)$. For example, consider the case $K=1$, we can obtain a closed form expression for the $r_1$ variable, and the optimization problem in \cref{eqn:op_perfect} can be rewritten as
\begin{equation}
r^*_1=\sum^M_{m=1}p_m\int^\infty_{-\infty}a(t)\exp(i(\theta^*_1-\lambda_m)t)\rd t,
\end{equation}
and
\begin{equation}
\begin{aligned}
\theta^*_1&=\argmax_{\theta\in\mathbb{R}}\left|\sum^M_{m=1}p_m\int^\infty_{-\infty}a(t)\exp(i(\theta-\lambda_m)t)\rd t\right|^2=\argmax_{\theta\in\mathbb{R}}\left|\sum^M_{m=1}p_mF(\theta-\lambda_m)\right|\,.
\end{aligned}
\end{equation}
Since $F(x)$ is a function that concentrates around $x=0$, we expect that $\theta^*_1\approx \lambda_{m^*}$, where $m^*=\argmax_mp_m$. 

In practice, we can only generate an approximate value of $\left\langle\psi\right|\exp(-it H)\ket{\psi}$ up to statistical errors for a finite number of time steps $t$. As a result, the loss function we use in our optimization problem is only an approximation of the ideal loss function presented in \cref{eqn:loss_multi_modal_perfect}. The main objective of this work is to generalize this idea to the estimation of multiple eigenvalues, and to rigorously implement this idea and control the complexity.

Using the concentration property of $F(x)$, we can demonstrate that the solution of the optimization problem is within $\delta/T$ of the dominant eigenvalues $\lambda_m$. Therefore by choosing $T_{\max}=\delta/\epsilon$, the parameter $\delta$ directly affects the circuit depth, and can be selected to be arbitrarily small if $R^{(K)}$ is small enough. Additionally, the solution of the optimization problem presented in \cref{eqn:op_perfect} is robust to the measurement noise, which means that we do not need to generate a large number of data points and samples to construct the approximated optimization problem. This ensures that the algorithm can achieve the Heisenberg limit, i.e., $T_{\mathrm{total}}=\Or(1/\epsilon)$ (see \cref{sec:informal_result}).

\subsection{Related work}\label{sec:related}

If the initial quantum state $\ket{\psi}$ is a precise eigenvector of $H$ with a single dominant eigenvalue, such as $\lambda_1$, the Hadamard test algorithm is the simplest algorithm for estimating this eigenvalue with $\epsilon$-accuracy. Given any real number $T\neq 0$, we can run the Hadmard test circuits (\cref{fig:qc}) several times to generate an estimate $Z_T\approx \left\langle\psi_1\right|\exp(-iT H)\ket{\psi_1}=\exp(-iT\lambda_1)$. We then define the approximation $\lambda^*_1=-\frac{\mathrm{arg} Z_T}{T}$. Because we assume $
\lambda_1\in[-\pi,\pi]$, it suffices to choose $T=1$ and $\lambda^*_1$ can be  arbitrarily close to $\lambda_1$ by increasing the number of measurements. This observation implies that the maximal running time $T
_{\max}=1$ suffices to achieve arbitrary precision when employing the Hadamard test algorithm. Although the Hadamard test algorithm only utilizes a short circuit depth, it has several limitations: (1) the initial state $\ket{\psi}$ must be an exact eigenstate, (2) the total runtime $T_{\mathrm{total}}$ does not satisfy the Heisenberg-limited scaling and is $\Omega(\epsilon^{-2})$, and (3) it cannot estimate multiple eigenvalues. The first two limitations can be addressed by using quantum phase estimation (QPE) and its variants \cite{KitaevShenVyalyi2002,PhysRevA.75.012328,NagajWocjanZhang2009,PhysRevLett.103.220502,BerryHiggins2009,HigginsBerryEtAl2007,GriffithsNiu1996}. The textbook version of QPE~\cite{NielsenChuang2000} uses multiple ancilla
qubits to store the phase and relies on inverse quantum Fourier transform (QFT). However, it can be difficult to employ the textbook version of QPE for multiple eigenvalue phase estimation due to the difficulty in handling spurious eigenvalues. In addition, although QPE can start from an imperfect eigenstate and saturates the Heisenberg-limited scaling, it requires the maximal runtime $T_{\max}$ at least $\pi/\epsilon$ to achieve $\epsilon$ accuracy~\cite[Section 5.2.1]{NielsenChuang2000}.

Ref.~\cite{Dutkiewicz2022heisenberglimited} recently proposed a multiple eigenvalue phase estimation method that satisfies the Heisenberg limit scaling without any spectral gap assumptions.  However, the theoretical analysis assumes that all non-dominant modes vanish~\cite[Definition 3.1 and Theorem 4.5]{Dutkiewicz2022heisenberglimited},  i.e., $p_{m'}=0$ for any $m'\in\mathcal{D}^c$. Additionally, it has not yet been demonstrated that the method satisfies the short-depth property (4) described at the beginning of the paper.

An alternative way to understand the phase estimation problem is to view it as a sparse Fourier transform problem in classical signal processing. In recent years, the sparse Fourier transform problem has been extensively studied in both the discrete~\cite{6576276,https://doi.org/10.1002/cpa.21455,7178651,DUARTE2013111,10.1007/978-3-642-32512-0_6,10.1145/2897518.2897650} and continuous~\cite{price_2015,10.1145/3406325.3451078,doi:10.1137/1.9781611977554.ch176} settings. In particular, Ref.~\cite{price_2015} concerns the recovery of the Fourier frequency of a continuous signal in the low noise setting. Our Theorem \ref{thm:gaussian} shares similarities with \cite[Theorem 1.1]{price_2015}. However, in our approach, each data point is generated by running the Hadamard test \emph{once}, resulting in significantly higher noise level than that  in \cite{price_2015}.  Our optimization based method may also be easier to implement in practice. Very recently, Ref.~\cite{ni2023lowdepth_2} introduces a robust multiple-phase estimation (RMPE) algorithm that can be considered as an extension of the multi-level robust phase estimation method (RPE) proposed in an earlier work \cite{ni2023lowdepth}. By leveraging a multi-level signal processing technique \cite{li2023note}, the RMPE algorithm can estimate multiple eigenvalues with Heisenberg-limited scaling without spectral gap assumptions. Under an additional spectral gap assumption, and with a different technique called ESPRIT~\cite{9000636} originally developed for super-resolution of signals and capable of efficiently estimating multiple eigenvalues in polynomial time~\cite{Stroeks_2022}, the resulting algorithm can satisfy the short-circuit depth property when the residual overlap $R^{(K)}/(\min_{m\in\mathcal{D}}p_m)$ is small.

Quantum subspace diagonalization (QSD) methods, quantum Krylov methods, and
matrix pencil methods~\cite{Cortes2021QuantumKS,Huggins_2020,PRXQuantum.3.020323,PhysRevA.95.042308,Motta_2020,Parrish2019QuantumFD,PRXQuantum.2.010333,doi:10.1021/acs.jctc.9b01125,OBrienTarasinskiTerhal2019} estimates the eigenvalues by solving certain projected  eigenvalue problems or singular value problems, and have been used to estimate ground-state and excited-state energies in a number of scenarios. 
Classical perturbation theories suggest that  such projected problems can be highly ill-conditioned and the solutions are sensitive to noise. However, it has been observed that these methods can perform better than the pessimistic theoretical predictions. Recently, Ref.~\cite{doi:10.1137/21M145954X} provided a theoretical analysis explaining this phenomenon in the context of quantum subspace diagonalization methods used for computing the smallest eigenvalue. Nevertheless, the error of these methods generally increase with respect to the number of data points, which is related to the dimension of the projected matrix problem. In contrast, the error of our algorithm  \textit{decreases} with respect to the number of data points, making it much more robust to measurement noise. It is worth noting that the sharp complexity analysis of these methods (especially, with respect to the dimension of the projected matrix problem) remains open.

The loss function \eqref{eqn:loss_multi_modal_perfect} utilized in this work is inspired by a recently developed subroutine called quantum complex exponential least squares (QCELS) \cite{DingLin2023}. QCELS employs the same quantum circuit as depicted in \cref{fig:qc} to generate data and uses an optimization-based approach for estimating a single dominant eigenvalue. 
Although our method is also based on solving an optimization problem, it is important for us to emphasize that directly extending QCELS to estimate multiple eigenvalues \textbf{does not} yield satisfactory results. We detail the differences in the following:
\begin{itemize}
    \item QCELS samples the time steps $t$ on a uniform grid as $t_n=\frac{n}{N}T$ for $0\leq n\leq N-1$. In order to satisfy the Heisenberg-limited scaling, Ref. \cite{DingLin2023} proposed a multi-level QCELS algorithm that gradually increases the step size $\tau=T/N$ as one refines the estimate to the eigenvalue.
However, such an approach can result in aliasing problems for multiple eigenvalue estimation since it may not be possible to distinguish $\exp(-i\lambda_{m_1}t_n)$ and $\exp(-i\lambda_{m_2}t_n)$ in the loss function if $(\lambda_{m_1}-\lambda_{m_2})T/N=2\pi q$ for some $m_1\neq m_2$ and $q\in \mathbb{Z}$. To overcome this difficulty, we combine the random sampling strategy~\cite{LinTong2022,SubasiSommaOrsucci2019} with QCELS and sample $t_n$'s randomly from a probability density $a(t)$. 
\item For estimating a single dominant eigenvalue, our method also achieves improved theoretical results. In Ref. \cite{DingLin2023}, QCELS requires $p_1>0.71$ and the circuit depth satisfies $T_{\max}=\delta/\epsilon$, where $\delta=\Theta(\sqrt{1-p_1})$. By applying our method to the single eigenvalue estimation problem, our analysis demonstrates that the method works for any $p_1>0.5$, and the constant $\delta$ can be improved to $\delta=\Theta(1-p_1)$. This improvement tightens the bound for the circuit depth. For clarity, the disparities in the constant dependence can be observed in \cref{fig:dt}.
\begin{figure}[H]
\centering
\begin{center}
\includegraphics[width=0.6\textwidth]{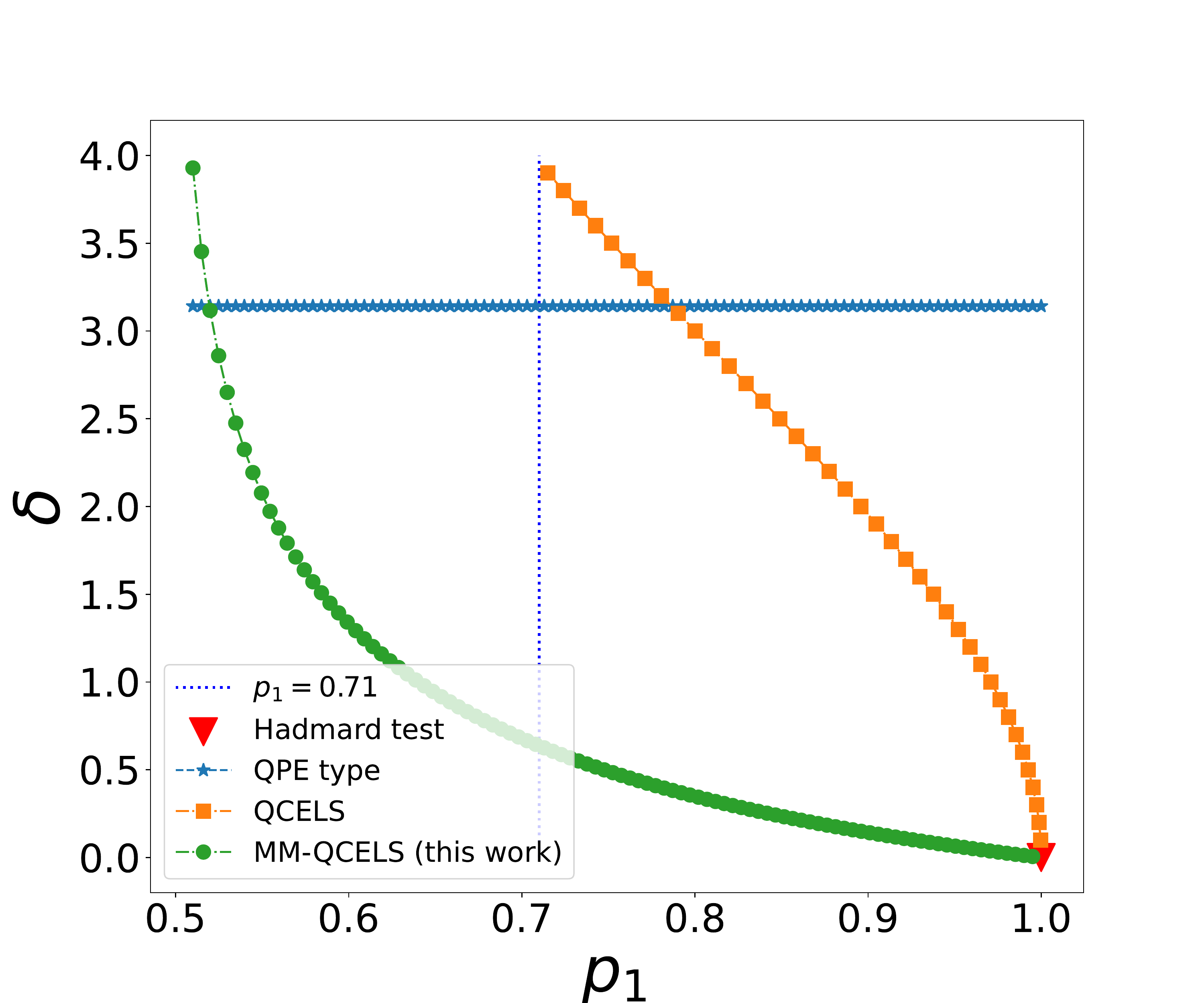}
\end{center}
\caption{A comparison
of the theoretical upper bound of $\delta=T_{\max}\epsilon$ for QCELS, QPE-type algorithm, and MM-QCELS for ground state energy estimation. The Hadamard test is only applicable when $p_1=1$ and the theoretical result in \cite{DingLin2023} requires $p_1>0.71$ for QCELS.}
\label{fig:dt}
\end{figure}
\item 
As previously noted, QCELS~\cite{DingLin2023} is formulated as a multi-level algorithm to attain Heisenberg-limited scaling and mitigate the aliasing problem. In our work, we also develop MM-QCELS as a multi-level method to help us control the random noise throughout the entire optimization domain. However, unlike QCELS~\cite{DingLin2023}, our approach can be simplified to solve the optimization problem with suitable parameters only once to approximate the dominant eigenvalues, resulting in a one-level algorithm. We put detailed discussion in the first comment of Remark \ref{re:add}.
\item In the original QCELS \cite{DingLin2023}, it was demonstrated that the maximum runtime can be further decreased to $T_{\max}=\wt{\Theta}(1/(\lambda_2-\lambda_1))$ by initially employing a ground state filter~\cite{LinTong2022} to enhance $p_1$. Our approach extends this characteristic to the estimation of multiple eigenvalues without requiring the use of a ground-state filter. This significantly simplifies the post-processing procedure. See the second comment of Remark \ref{re:add}.
\end{itemize}

We include a summary of the comparison among various methods for estimating multiple eigenvalues in Table \ref{table:1}.
\\
\begin{table}[h!]
\centering
\begin{tabular}{c|c c c c c} 
 \hline
 \hline Algorithms & \multicolumn{5}{c}{Requirement}\\  
 & (1) Nonzero& (2) Heisenberg& (3)Single & (4) Short& (5) Use\\  
 & Residual & limit& ancilla& depth& gap\\
 \hline
 QSD type~(\cite{Cortes2021QuantumKS}-\cite{doi:10.1137/21M145954X})&\cmark&?&?&?&?\\
 \hline
 QEEA~\cite{Somma2019}&\cmark&\xmark&\cmark&?&?\\
 \hline
 ESPRIT~\cite{Stroeks_2022}&?&?&\cmark&?&?\\
 \hline
 Ref.~\cite{Dutkiewicz2022heisenberglimited}  & ? & \cmark &\cmark & ?&\xmark\\
 \hline
 Ref.~\cite[Theorem III.5]{ni2023lowdepth_2}  & \cmark & \cmark &\cmark & \xmark& \xmark\\
 \hline
  Ref.~\cite[Theorem IV.2]{ni2023lowdepth_2} & \cmark & \cmark &\cmark & \cmark& ?\\
 \hline
 MM-QCELS (this work)& \cmark & \cmark &\cmark & \cmark& \cmark\\
 \hline
 \hline
\end{tabular}
\caption{Comparison of various methods for estimating multiple eigenvalues. 
}
\label{table:1}
\end{table}

\section{Main method and informal complexity result}\label{sec:main_method}
This section begins by presenting a quantum circuit for data generation. Subsequently, we use these data points to formulate an approximated optimization problem and propose the main algorithm. Finally, we provide an intuitive complexity analysis for our algorithm, followed by its rigorous statement.

\subsection{Generating  data from quantum circuit}\label{sec:mqc}

The quantum circuit used in this paper is the same as that used in the Hadamard test. In \cref{fig:qc}, we may
\begin{itemize}
    \item[(1)] Set $W=I$, measure the ancilla qubit and define a random variable $X_n$ such that $X_n=1$ if the outcome is $0$ and $X_n=-1$ if the outcome is $1$. Then 
\begin{equation}\label{eqn:X}
\mathbb{E}(X_n)=\mathrm{Re}\left(\left\langle\psi\right|\exp(-i t H)\ket{\psi}\right)\,.
\end{equation}

\item[(2)] Set $W=S^\dagger$ ($S$ is the phase gate), measure the ancilla qubit and define a random variable $Y_n$ such that $Y_n=1$ if the outcome is $0$ and $Y_n=-1$ if the outcome is $1$. Then 
\begin{equation}\label{eqn:Y}
\mathbb{E}(Y_n)=\mathrm{Im}\left(\left\langle\psi\right|\exp(-i t H)\ket{\psi}\right)\,.
\end{equation}
\end{itemize}

We assume an oracle access to the Hamiltonian evolution operator $e^{-i t H}$ for any $t\in\mathbb{R}$. This assumption is not overly restrictive. For instance, if we can approximate $e^{-i \tau H}$ for some desired $\tau>0$ using the Trotter method of a certain order (the choice of $\tau$ and the order of the Trotter splitting is problem dependent and should balance between efficiency and accuracy), we may express $e^{-i t H}=\left(e^{-i \tau H}\right)^p e^{-i \tau' H}$, where $t=p\tau+\tau'$, $p\in\ZZ$, and $\abs{\tau'}<\tau$. The cost for implementing $e^{-i \tau' H}$ should be no larger than that of implementing $e^{-i \tau H}$. The Trotter error can be systematically controlled and this is an orthogonal issue to the type of error considered in this paper (see, e.g., \cite[Appendix D]{LinTong2022} for how to factor such errors into the analysis). Therefore, throughout the paper, we assume that the implementation of $e^{-i t H}$ for any $t\in\mathbb{R}$ is exact.

Given a set of time points $\{t_n\}^N_{n=1}$ drawn from the probability density $a(t)$, we use the quantum circuit in \cref{fig:qc} to prepare the following data set:
\begin{equation}\label{eqn:dataset}
    \mathcal{D}_{H}=\left\{\left(t_n,Z_n\right)\right\}^{N}_{n=1}\,.
\end{equation}
By running the quantum circuit (\cref{fig:qc}) for each $t_n$ \textit{once}, we define the output
\begin{equation}\label{eqn:Z_n}
Z_n=X_{n}+iY_{n}.
\end{equation}
Here $X_{n},Y_{n}$ are independently generated by the quantum circuit (\cref{fig:qc}) with different $W$ and satisfy \eqref{eqn:X}, \eqref{eqn:Y} respectively. Hence, we have
\begin{equation}\label{eqn:Zn_expect}
\mathbb{E}(Z_n)=\left\langle\psi\right|\exp(-i t_n H)\ket{\psi},\quad |Z_n|\leq 2\,.
\end{equation}
Therefore $Z_n$ is an unbiased estimation $\left\langle\psi\right|\exp(-i t_n H)\ket{\psi}=\sum^{M}_{m=1}p_m\exp(-i \lambda_m t_n)$.
We also note that if we use the above method to prepare the data set in \cref{eqn:dataset}, the maximal simulation time is $T_{\max}=\max_{1\leq n\leq N}|t_n|$ and the total simulation time is $\sum^N_{n=1}|t_n|$. 

Given $t_n$, we can sample the circuit $N_s$ times such that $\abs{Z_n-\left\langle\psi\right|\exp(-i t_n H)\ket{\psi}}=\Or(1/\sqrt{N_s})$. Most algorithms, including the QCELS algorithm in Ref.~\cite{DingLin2023}, require this repetition step. In this aspect, our algorithm is innovative in that it achieves convergence with just a single sample per $t_n$, i.e., $N_s=1$. For simplicity, we assume $N_s=1$ throughout the paper. Increasing $N_s$ reduces the statistical noise in $Z_n$ and can further decrease the error in the eigenvalue estimate. This also increases $T_{\mathrm{total}}$ by a factor of $N_s$ without affecting $T_{\max}$.

\begin{figure}[H]
\centering
\begin{center}
\includegraphics[width=0.6\textwidth]{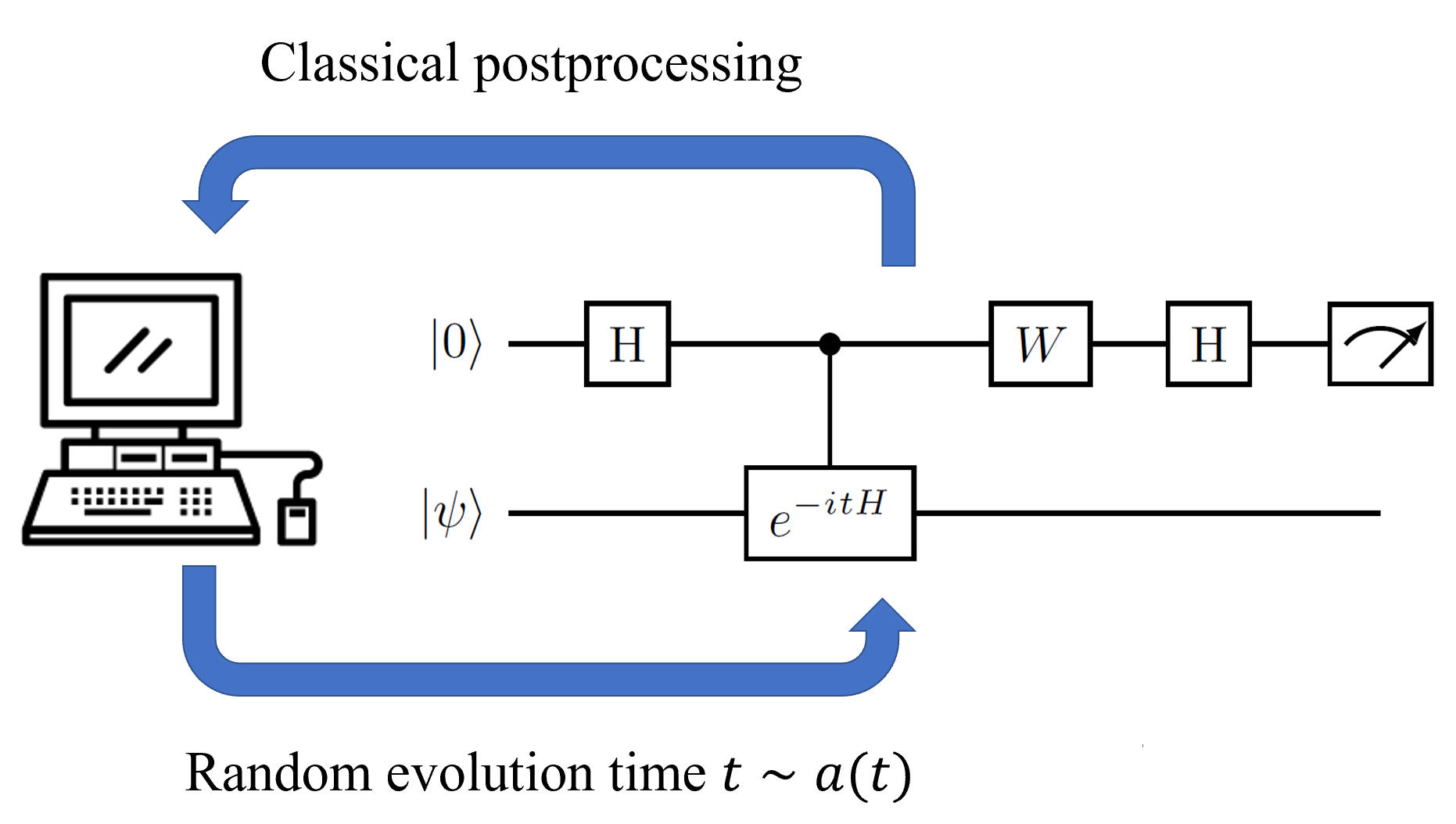}
\end{center}
\caption{Choosing $W=I$ or $W=S^\dagger$ ($S$ is the phase gate), the Hadamard test circuit allows us to estimate the real or the imaginary part of $\braket{\psi|\exp(-itH)|\psi}$ on quantum computers. $\mathrm{H}$ is the Hadamard gate. The classical computer provides the evolution time $t$ according to a probability density $a(t)$, and performs postprocessing on the quantum data for eigenvalue estimation. }
\label{fig:qc}
\end{figure}

Algorithm \ref{alg:data} describes the data generating process.

\begin{breakablealgorithm}
      \caption{Data generator}
  \label{alg:data}
  \begin{algorithmic}[1]
  \State \textbf{Preparation:} Number of data pairs: $N$; probability density: $a(t)$; 
 \State $n\gets 1$;
  \While{$n\leq N$}
  \State Generate a random variable $t_n$ with the probability density $a(t)$.
  \State Run the quantum circuit (Figure \ref{fig:qc}) with $t=t_n$ and $W=I$ to obtain $X_{n}$.
  \State Run the quantum circuit (Figure \ref{fig:qc}) with $t=t_n$ and $W=S^\dagger$ to obtain $Y_{n}$.
  \State $Z_{n}\gets X_{n}+i Y_{n}$.
  \State $n\gets n+1$
  \EndWhile
    \State \textbf{Output:} $\left\{(t_n,Z_{n})\right\}^N_{n=1}$
    \end{algorithmic}
\end{breakablealgorithm}

\subsection{Main algorithm}
After generating the data set, we define the numerical loss function (referred to as the loss function for short) as
\begin{equation}\label{eqn:loss_multi_modal}
L_{K}\left(\{r_k\}^{K}_{k=1},\{\theta_k\}^{K}_{k=1}\right)=\frac{1}{N}\sum^N_{n=1}\left|Z_n-\sum^{K}_{k=1}r_k\exp(-i\theta_k t_n)\right|^2
\end{equation}
and the optimization problem
\begin{equation}\label{eqn:op}
\left(\{r^*_k\}^{K}_{k=1},\{\theta^*_k\}^{K}_{k=1}\right)=\argmin_{r_k\in\CC,\theta_k\in\mathbb{R}}L_{K}\left(\{r_k\}^{K}_{k=1},\{\theta_k\}^{K}_{k=1}\right)\,.
\end{equation}
Compared to the ideal loss function described in  \cref{eqn:loss_multi_modal_perfect}, the current loss function is considerably noisier. Define the expectation error $E_n=Z_n-\sum^{M}_{m=1}p_m\exp(-i \lambda_m t_n)$. We note that $|E_n|$ is bounded by 3 but not small since $Z_n$ is generated by running the quantum circuit only once. On the other hand, the expectation of $E_n$ is zero and $\{E_n\}^N_{n=1}$ are independent. This implies $\left|\frac{1}{N}\sum^N_{n=1}E_n\right|=\Or\left(1/\sqrt{N}\right)$. Define the vector $\vr=(r_1,r_2,\dots,r_K)$ and $\vtheta=(\theta_1,\theta_2,\dots,\theta_K)$, and plug this into \eqref{eqn:loss_multi_modal}, we have
\begin{equation}\label{eqn:first_informal}
\begin{aligned}
&\argmin_{\vr,\vtheta}L_{K}\left(\vr,\vtheta\right)\\
=&\argmin_{\vr,\vtheta}\frac{1}{N}\sum^N_{n=1}\left|Z_n-\sum^{K}_{k=1}r_k\exp(-i\theta_k t_n)\right|^2\\
=&\argmin_{\vr,\vtheta}\frac{1}{N}\sum^N_{n=1}\left|\sum^{M}_{m=1}p_m\exp(-i \lambda_m t_n)-\sum^{K}_{k=1}r_k\exp(-i\theta_k t_n)\right|^2\\
&-\frac{2}{N}\sum^N_{n=1}\mathrm{Re}\left\langle E_n,\sum^{K}_{k=1}r_k\exp(-i\theta_k t_n)\right\rangle\\
\approx &\argmin_{\vr,\vtheta}\int^\infty_{-\infty}a(t)\left|\sum^{M}_{m=1}p_m\exp(-i \lambda_m t)-\sum^{K}_{k=1}r_k\exp(-i\theta_k t)\right|^2\rd t\\
=&\argmin_{\vr,\vtheta}\mathcal{L}_{K}\left(\vr,\vtheta\right)\,,
\end{aligned}
\end{equation}
where $\left\langle a,b\right\rangle=\overline{a} \cdot b$ for any $a,b\in\mathbb{C}$. In the second equality, we omit  terms that are independent of $(\vr,\vtheta)$. In the approximation step, we use $\frac{2}{N}\sum^N_{n=1}\mathrm{Re}\left\langle E_n,\sum^{K}_{k=1}r_k\exp(-i\theta_k t_n)\right\rangle\approx 0$ when $N\gg 1$. While the loss function $L_K$ may not converge to the ideal loss function as $N$ approaches infinity due to the statistical noise in $Z_n$, we find that the optimization problem in \cref{eqn:op} can still yield a solution that approaches that in \cref{eqn:op_perfect} with the ideal loss.

After formulating the loss function \eqref{eqn:loss_multi_modal} and the optimization problem \eqref{eqn:op}, we are ready to introduce our main algorithm. Inspired by the multi-level QCELS~\cite{DingLin2023}, the algorithm contains two steps:  
\begin{itemize}
    \item \textbf{Step 1:} Choose a relative small $T_0$. Set $T=T_{0}$, $a(t)=a_{T_0}(t)$ and solve the optimization problem \eqref{eqn:op} to give a rough estimation for each dominant eigenvalues $\{\lambda_m\}_{m\in\mathcal{D}}$.  
    
    \item \textbf{Step 2:} Gradually increase $T$ so that: 
    \begin{enumerate}
    
    \item The solution with $T_{j}$ gives a good initial guess for the optimization problem with $T_{j+1}$.

    \item The total running time still satisfies the Heisenberg-limit.
    \end{enumerate}
\end{itemize}
The detailed algorithm is summarized in Algorithm \ref{alg:main}, which is referred to as the multi-modal, multi-level quantum complex exponential least squares (MM-QCELS) algorithm. We note that the optimization problems \eqref{eqn:op_step_2} in Algorithm \ref{alg:main} have the technical constraint $\|\vr\|_1\leq 1$. This constraint is natural since each $r_k$ should approximate some $p_{m_k}$, and the sum of the absolute value thus should exceed $1$. This constraint is added to simplify the complexity analysis for a uniform upper bound of the expectation error (see Lemma \ref{lem:expectation_error} in Appendix \ref{sec:bound_expectation_error}  for detail). In practice, we find that the performance of the algorithm is independent of this constraint (see the numerical tests in Section \ref{sec:ns}).

\begin{breakablealgorithm}
      \caption{Multi-modal, multi-level quantum complex exponential least squares (MM-QCELS)}
  \label{alg:main}
  \begin{algorithmic}[1]
  \State \textbf{Preparation:} Number of data pairs: $\{N_j\}^l_{j=0}$; number of iterations: $l$; sequence of time steps: $\{T_j\}^{l}_{j=0}$; sequence of probability distributions $\{a_{T_j}(t)\}^l_{j=0}$; number of dominant eigenvalues: $K$
  \State \textbf{Running:}
  \State $k\gets 1$; \Comment{Step 1 starts}
  \While{$k\leq K$}
  \State $\lambda_{\min,k}\gets-\pi$; $\lambda_{\max,k}\gets\pi$; \Comment{$[\lambda_{\min,k},\lambda_{\max,k}]$ is the estimation interval.}
  \State $k\gets k+1$
  \EndWhile
  \State $j\gets 0$; 
  \While{$j\leq l$}

  \State Generate a data set $\left\{\left(t_n,Z_{n}\right)\right\}^{N_j}_{n=1}$ using Algorithm \ref{alg:data} with $a_{T_j}(t)$.  
  \State Define the loss function $L_{K}(r,\theta)$ in \cref{eqn:loss_multi_modal}.  
  \State Minimizing loss function: 
  \begin{equation}\label{eqn:op_step_2}
  \left(\vr^*,\vtheta^*\right)\gets\argmin_{\|\vr\|_1\leq 1,\theta_k\in[\lambda_{\min,k},\lambda_{\max,k}]}L_{K}\left(\vr,\vtheta\right)\,.
  \end{equation}
  \State $\lambda_{\min,k}\gets\theta^*_{k}-\frac{\pi}{T_j}$;  $\lambda_{\max,k}\gets\theta^*_{k}+\frac{\pi}{T_j}$ \Comment{Shrink the search interval.}
  \State $j\gets j+1$
  \EndWhile \Comment{Step 2 ends}
  \State \textbf{Output:} $\{(\vr^*,\vtheta^*)\}^{K}_{k=1}$
    \end{algorithmic}
\end{breakablealgorithm}

\subsection{Complexity analysis of Algorithm \ref{alg:main}}\label{sec:informal_result}

In this section, we study the complexity of Algorithm \ref{alg:main}. First, the cost of the algorithm depends on a number of parameters used throughout the analysis, including

\begin{enumerate}
\item The minimal dominant overlap: 
\begin{equation}\label{eqn:pK_min}
p^{(K)}_{\min}=\min_{m\in\mathcal{D}}p_m\,;
\end{equation}
\item The minimal dominant spectral gap: 
\begin{equation}\label{eqn:DeltaK_lambda}
\Delta^{(K)}_\lambda=\min_{m,m'\in\mathcal{D},m\neq m'}\left|\lambda_{m'}-\lambda_m\right|\,;
\end{equation}
\item The residual overlap: \begin{equation}\label{eqn:RK}
R^{(K)}=\sum_{m'\in\mathcal{D}^c}p_{m'}\,.
\end{equation}
\end{enumerate}

Now, we are ready to introduce the  complexity result of Algorithm \ref{alg:main}, which is summarized in the following theorem.  

\begin{thm}\label{thm:gaussian}
Given the failure probability $0<\eta<1/2$, error $\epsilon>0$, and any $\zeta>0$. Assume $p^{(K)}_{\min}>3R^{(K)}$ and define $q=\Theta(R^{(K)}/p^{(K)}_{\min})$. In Algorithm \ref{alg:main}, we choose the following parameters:
    \begin{itemize}
    \item Set $\gamma =\Theta\left(\log\left(1/\min\left\{p^{(K)}_{\min}q,\left(p^{(K)}_{\min}\right)\left(p^{(K)}_{\min}-3R^{(K)}\right)\right\}\right)\right)$.
    \item For Step 1, set 
    \begin{equation}\label{eqn:condition_gaussian}  T_0=\widetilde{\Theta}\left(\left(\Delta^{(K)}_{\lambda}\right)^{-1}\log(q^{-1})\right),\quad N_0=\widetilde{\Theta}\left(T^2_0\left(p^{(K)}_{\min}\left(p^{(K)}_{\min}-3R^{(K)}\right)\right)^{-2}\polylog(\eta^{-1})\right)\,.
    \end{equation}
    \item For Step 2, set $l=\max\{\left\lceil\log_2\left(q/(\epsilon T_0)\right)\right\rceil,1\}$, $T_j=2^jT_0$, and 
    \begin{equation}\label{eqn:N_1}
N_j=\widetilde{\Theta}\left(\left(p^{(K)}_{\min}\right)^{-4}q^{-2-\zeta}\polylog(\log(\zeta^{-1})\eta^{-1})\right)
    \end{equation} for $1\leq j\leq l$.
    \end{itemize}
    This gives
    \begin{equation}
      T_{\max}=\frac{\delta}{\epsilon},\quad T_{\mathrm{total}}=\wt{\Theta}\left(\frac{1}{\left(p^{(K)}_{\min}\right)^4\delta^{1+\zeta}\epsilon}\polylog(\log(\zeta^{-1})\delta^{-1}\eta^{-1})\right)\,.
\end{equation}
where $\delta=\widetilde{\Theta}\left(q\log(q^{-1})\right)$. In the above equations, the constant in $\widetilde{\Theta}$ depends at most polynomially on $\log\left((p^{(K)}_{\min}(p^{(K)}_{\min}-3R^{(K)}))^{-1}K\right)$.
Then with probability $1-\eta$ ,  for any $m\in\mathcal{D}$, there exists a unique $1\leq k_m\leq K$ such that
    \begin{equation}\label{eqn:final_close_gaussian}
        |\lambda_{m}-\theta^*_{k_m}|\leq \epsilon\,.
    \end{equation}
Here, $\zeta$  can  be chosen arbitrarily close to $0$ and the constant in $\widetilde{\Theta}$ depends on $\zeta$.
    \end{thm}

We put the proof of the above theorem in the Appendix. To provide a clear exposition of the core concept, we first present the intuitive idea of the proof in Appendix \ref{sec:pf_intuitive}. The key step in our proof is to demonstrate that solving the ideal optimization problem~\eqref{eqn:op_perfect} could yield a precise approximation to dominant eigenvalues with a
short maximal running time. 
The rigorous proof of the theorem is then given in \cref{sec:proof_of_thm}. In particular, we rigorously bound the approximation error in \eqref{eqn:first_informal}, and optimizing the numerical loss function~\eqref{eqn:loss_multi_modal} gives us an accurate approximation of dominant eigenvalues with low cost.

\begin{rem}\label{re:add} The results of \cref{thm:gaussian} can be extended as follows:
\begin{enumerate}
    \item In the above theorem, we could relax the condition that $T_0=\widetilde{\Theta}((\Delta^K_{\lambda})^{-1}\log(q^{-1})$. To be more precise, if we are given a lower bound $\Delta_{\mathrm{low}}$ for $\Delta^K_{\lambda}$, we could set $T_0=\widetilde{\Theta}((\Delta_{\mathrm{low}})^{-1}\log(q^{-1}))$, and the result of the theorem still holds.
    
    \item  In \cref{thm:gaussian} and Algorithm \ref{alg:main}, we utilize a sequence of $T_n$ to approximate the dominant eigenvalues. This is due to a technical consideration in our theoretical analysis. Specifically, in order to ensure the feasibility of optimization problem \eqref{eqn:op_step_2}, our proof requires that the random discrepancy between \eqref{eqn:op_step_2} and the ideal optimization problem \eqref{eqn:op_perfect} remains uniformly small when $\theta_k\in [\lambda_{\min,k},\lambda_{\max,k}]$ (see Appendix \ref{sec:proof_of_thm} \cref{lem:single_op_problem} for detail). Achieving this requires a uniform bound for infinitely many continuous random variables, which necessitates constraining $\lambda_{\max,k}-\lambda_{\min,k}$ to be sufficiently small, thus guaranteeing the manageability of random noise with a reasonable number of samples. Detailed information can be found in \cref{sec:bound_expectation_error}.

In practical calculations, Algorithm \ref{alg:main} may be simplified into a one-level algorithm by directly choosing sufficiently large values for $T_{\max}$ and $N$. Moreover, if we optimize our function only over a finite number of grid intervals such that $\theta_k\in {-\pi+k\epsilon,0\leq k<\frac{2\pi}{\epsilon}}$, it is theoretically possible to show that by minimizing \eqref{eqn:op_step_2} once with $T_{\max}=\widetilde{\Theta}\left(q/\epsilon\right)$ and $N=\widetilde{\Theta}(q^{-2-o(1)})$, we can achieve $\epsilon$-accuracy for all dominant eigenvalues. This approach also offers benefits such as short circuit depth and Heisenberg-limited scaling.
\item When a spectral gap exists between the dominant eigenvalues and the non-dominant ones (represented as $\Delta_{\lambda}$), it may be feasible to further reduce the maximum runtime to $T_{\max}=\wt{\Theta}(1/\min\{\Delta_{\lambda},\Delta^{(K)}_{\lambda}\})$. The reason is given in the first point of \cref{rem:after_intuition}.
\end{enumerate}
\end{rem}

\section{Numerical results}\label{sec:ns}
In this section, we numerically demonstrate the efficiency of our method using two different models. In \cref{sec:Ising}, we compare the performance of Algorithm \ref{alg:main} with QPE (textbook version~\cite{NielsenChuang2000}) for a transverse-field Ising model. In \cref{sec:Hubbard}, we compare the performance of Algorithm \ref{alg:main} with QPE for a Hubbard model. In both cases, we assume there are two dominant eigenvalues ($\lambda_1,\lambda_2$), meaning $K=2$. We share the code on Github (\url{https://github.com/zhiyanding/MM-QCELS}).


In our numerical experiments, we normalize the spectrum of the original Hamiltonian $H$ so that the eigenvalues belong to $[-\pi/4,\pi/4]$. Given a Hamiltonian $H$, we use the normalized Hamiltonian in the experiment:
\begin{equation}\label{eqn:normalize_H}
\widetilde{H}=\frac{\pi H}{4\|H\|_2}\,.
\end{equation}
Recall that the textbook version of QPE~\cite{NielsenChuang2000} relies on a quantum circuit that involves the inverse Quantum Fourier Transform (QFT). This circuit serves to encode the information of approximate eigenvalues using the ancilla register. By measuring the ancilla qubits, we could acquire an output $\theta$ that closely approximates $\lambda_k$, one of the eigenvalues of Hamiltonian $H$. To find an approximation to the smallest eigenvalue, we repeat the quantum circuit for $N$ times and defines the approximation $\theta^*_1=\min_{1\leq n\leq N} \theta_n$, where $\theta_n$ is the output of $n$-th repetition. The analysis of this method can be found in e.g., ~\cite[Section I.A]{LinTong2022}. However, QPE can also produce spurious eigenvalues which lead to failures, and it is not straightforward to generalize the procedure above for estimating multiple eigenvalues.

Consequently, in our experiment, we first utilize QPE to estimate the smallest eigenvalue $\lambda_1$ and measure the error accordingly. We then use Algorithm \ref{alg:main} to estimate the two dominant eigenvalues and measure the error by (assuming $\theta^*_1\leq \theta^*_2$)
\begin{equation}\label{eqn:max_error}
\mathrm{error}=\max\{|\theta^*_1-\lambda_1|,|\theta^*_2-\lambda_2|\}.
\end{equation}
For simplicity, in this section, Algorithm \ref{alg:main} is implemented without the constraint $\|\vr\|_1<1$ in \eqref{eqn:op_step_2}. While QPE's error is gauged based on a single eigenvalue, the error of Algorithm \ref{alg:main} is evaluated using the maximum error across two eigenvalues. This testing design intrinsically gives QPE a head start. Even with this bias, we demonstrate that Algorithm \ref{alg:main}  can outperform QPE.

\subsection{Ising model}\label{sec:Ising}
Consider the one-dimensional transverse field Ising model (TFIM) model defined on $L$ sites with periodic boundary conditions:
\begin{equation}\label{eqn:H_Ising}
H=-\left(\sum^{L-1}_{i=1} Z_{i}Z_{i+1}+Z_{L}Z_1\right) -g\sum^L_{i=1} X_i
\end{equation}
where $g$ is the coupling coefficient, $Z_i,X_i$ are Pauli operators for the $i$-th site and the dimension of $H$ is $2^L$.
We choose $L=8,g=4$. In the test, we set $p_1=p_2=0.4$.  In Algorithm \ref{alg:main}, the parameters are set to be $K=2,T_0=2(\lambda_2-\lambda_1)^{-1},N_0=3\times 10^3,N_{j\geq 1}=2\times 10^3$ and $\gamma=1$. An illustrative example, Fig \ref{fig:TFIM_loss} demonstrates the landscape of the loss function in \cref{eqn:loss_multi_modal}. As $T$ increases, the landscape of the loss function becomes increasingly complex. However, the value of the loss function around the true eigenvalues decreases significantly, which also leads to a reduction in run-to-run variation of the optimizer. As a result, the optimizer concentrates more tightly around the true eigenvalues.

We apply Algorithm \ref{alg:main} and QPE to estimate the dominant eigenvalues ($\lambda_1,\lambda_2$) of the normalized Hamiltonian $\widetilde{H}$ according to \cref{eqn:normalize_H}.  We then run Algorithm \ref{alg:main} (with $K=2,T_0=2(\lambda_2-\lambda_1)^{-1},N_0=3\times 10^3,N_{j\geq 1}=2\times 10^3$ and $\gamma=1$) to compute the error of both $\lambda_1$ and $\lambda_2$. We also run QPE $10$ times \textit{only} to estimate $\lambda_1$. The comparison of the results is shown in \cref{fig:Ising}. We find that the errors of both QPE and Algorithm \ref{alg:main} are proportional to the inverse of the maximal evolution time ($T_{\max}$). But the constant factor $\delta=T\epsilon$ of Algorithm \ref{alg:main} is much smaller than that of QPE. \cref{fig:Ising} shows that Algorithm \ref{alg:main} reduces the maximal evolution time by two orders of magnitude. The error of QPE is observed to scale as $6\pi/T$. Moreover, the total evolution time ($T_{\mathrm{total}})$ of Algorithm \ref{alg:main} is also slightly smaller than that of QPE. 

\begin{figure}[htbp]
     \subfloat{
         \centering
         \includegraphics[width=0.48\textwidth,height=0.3\textheight]{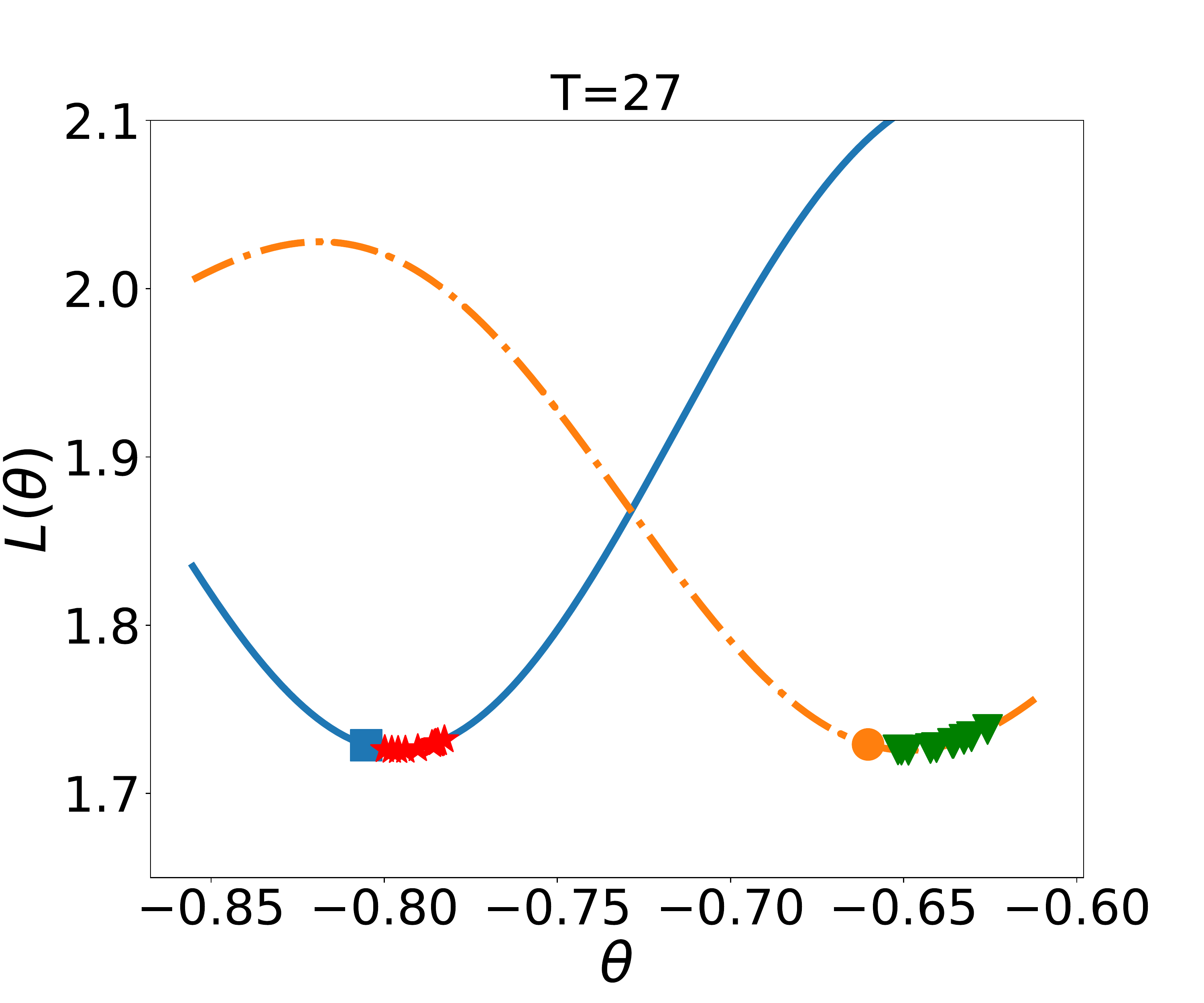}
     }
     \subfloat{
         \centering
         \includegraphics[width=0.48\textwidth,height=0.3\textheight]{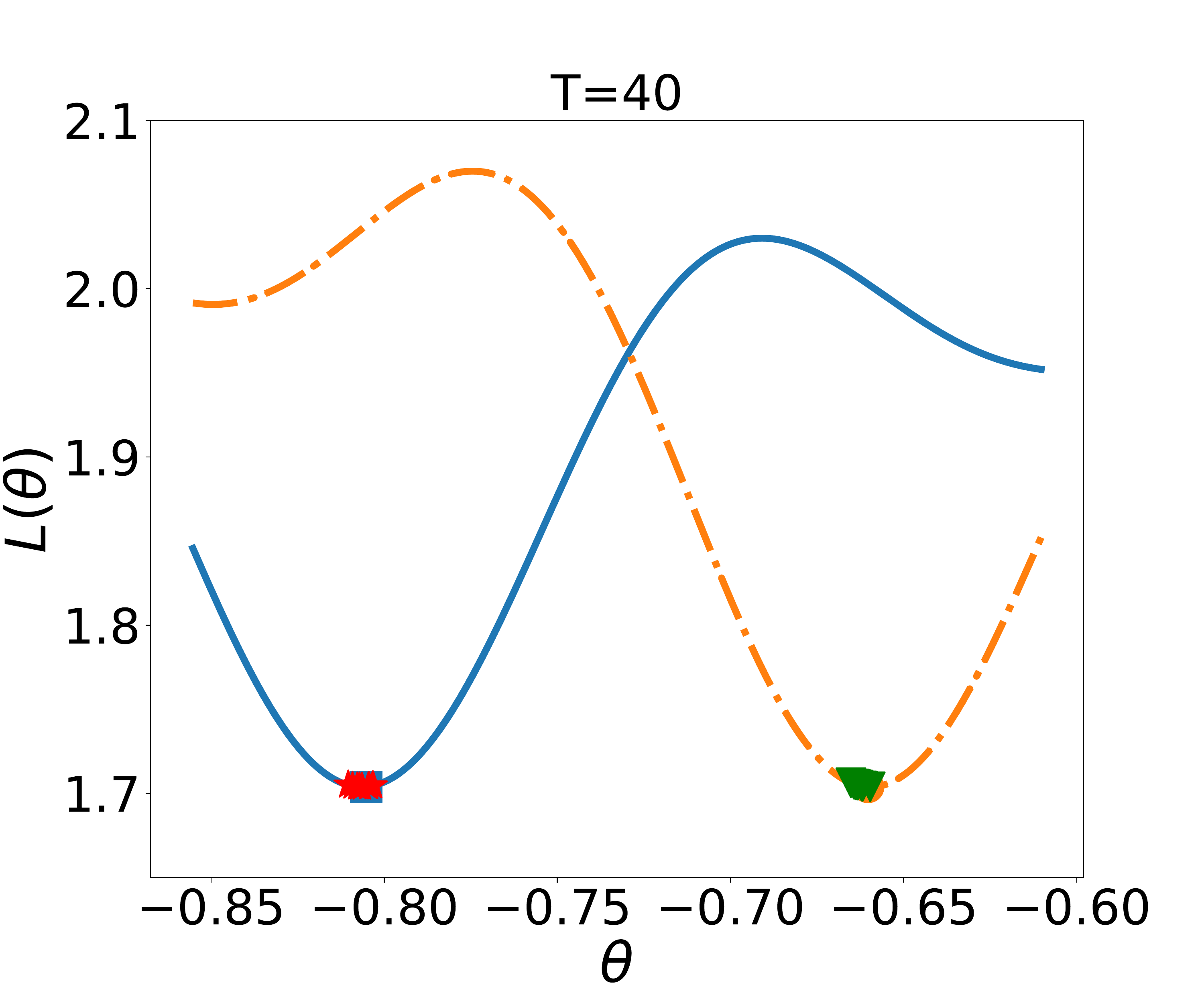}
     }
     \\
     \subfloat{
         \centering
         \includegraphics[width=0.48\textwidth,height=0.3\textheight]{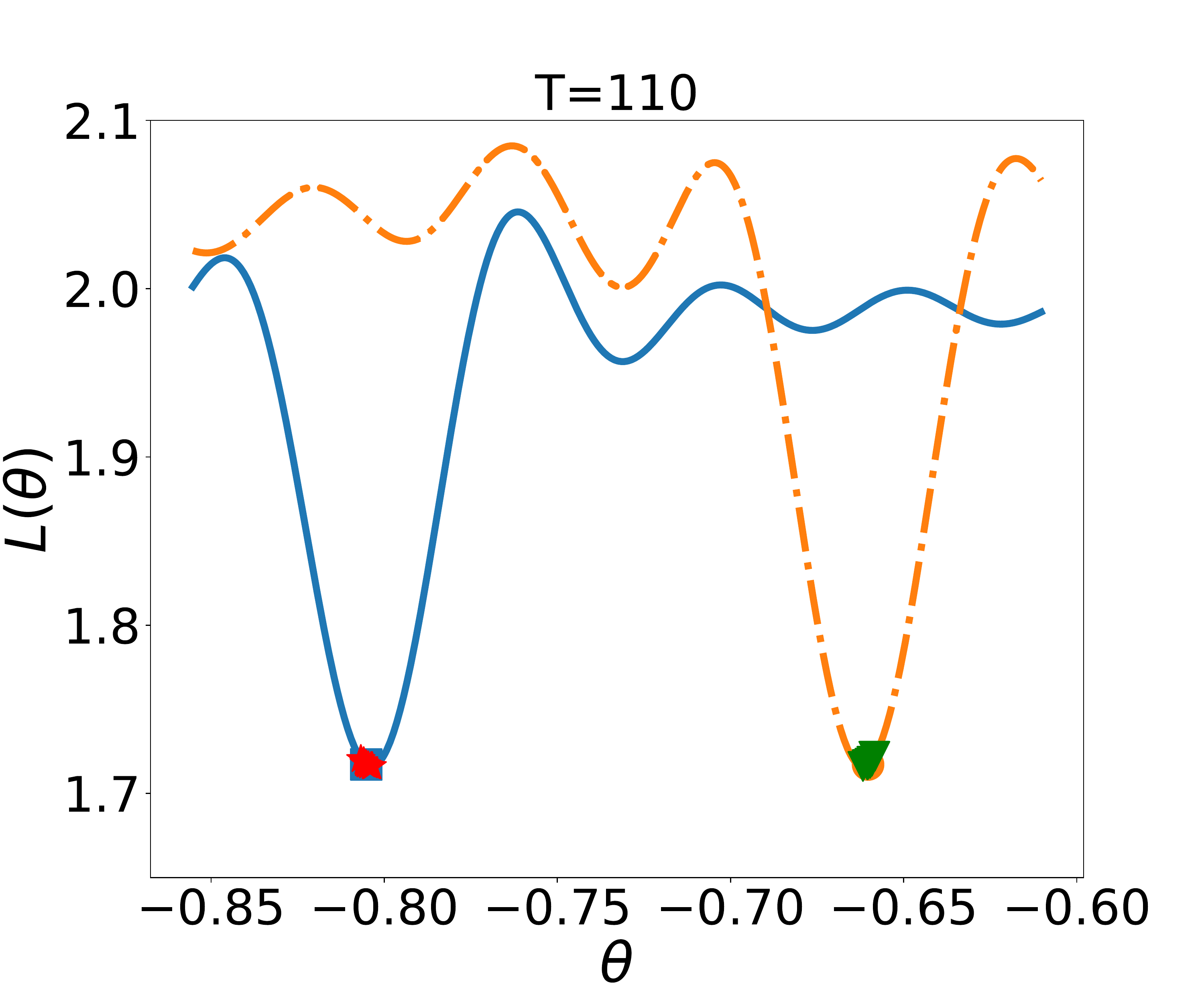}
     }
     \subfloat{
         \centering
         \includegraphics[width=0.48\textwidth,height=0.3\textheight]{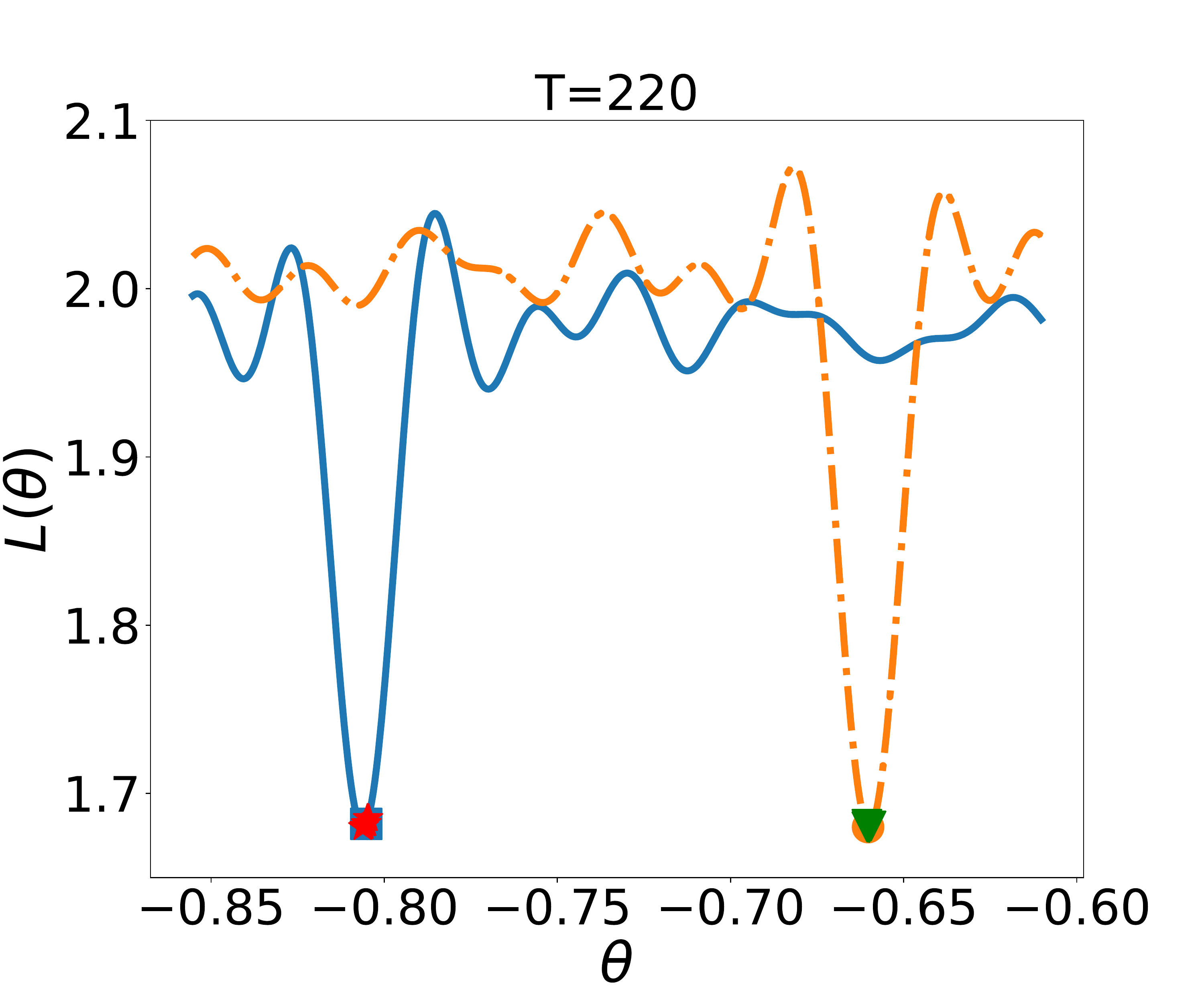}
     }
     \caption{
     \label{fig:TFIM_loss} The landscape of the loss function $L$ \eqref{eqn:loss_multi_modal} from the TFIM model at time $T_l$ with $p_1=0.4$, $p_2=0.4$, $K=2$, $T_0=2(\lambda_2-\lambda_1)^{-1}$, $N_0=3\times 10^3$, $N_{j\geq 1}=2\times 10^3$, $\gamma=1$. We run the experiment ten times and plot the positions of ten minimizers $\vtheta^*$ using star ($\theta^*_1$) and triangle ($\theta^*_2$) points. The landscape of the loss function is calculated using the data from the last experiment. The blue solid line stands for $L_K(\vr^*,\theta_1,\theta_2^*)$ (the variable is $\theta_1$), and the blue square is the true eigenvalue $\lambda_1$. The orange dash--dotted stands for $L_{K}(\vr^*,\theta_1^*,\theta_2)$ (the variable is $\theta_2$), and the orange filled circle is the true eigenvalue $\lambda_2$. When $T$ is large, the minimizer of the loss function concentrates around the dominant eigenvalues $\lambda_1$ and $\lambda_2$.}
\end{figure}

\begin{figure}[htbp]
     \subfloat{
         \centering
         \includegraphics[width=0.48\textwidth]{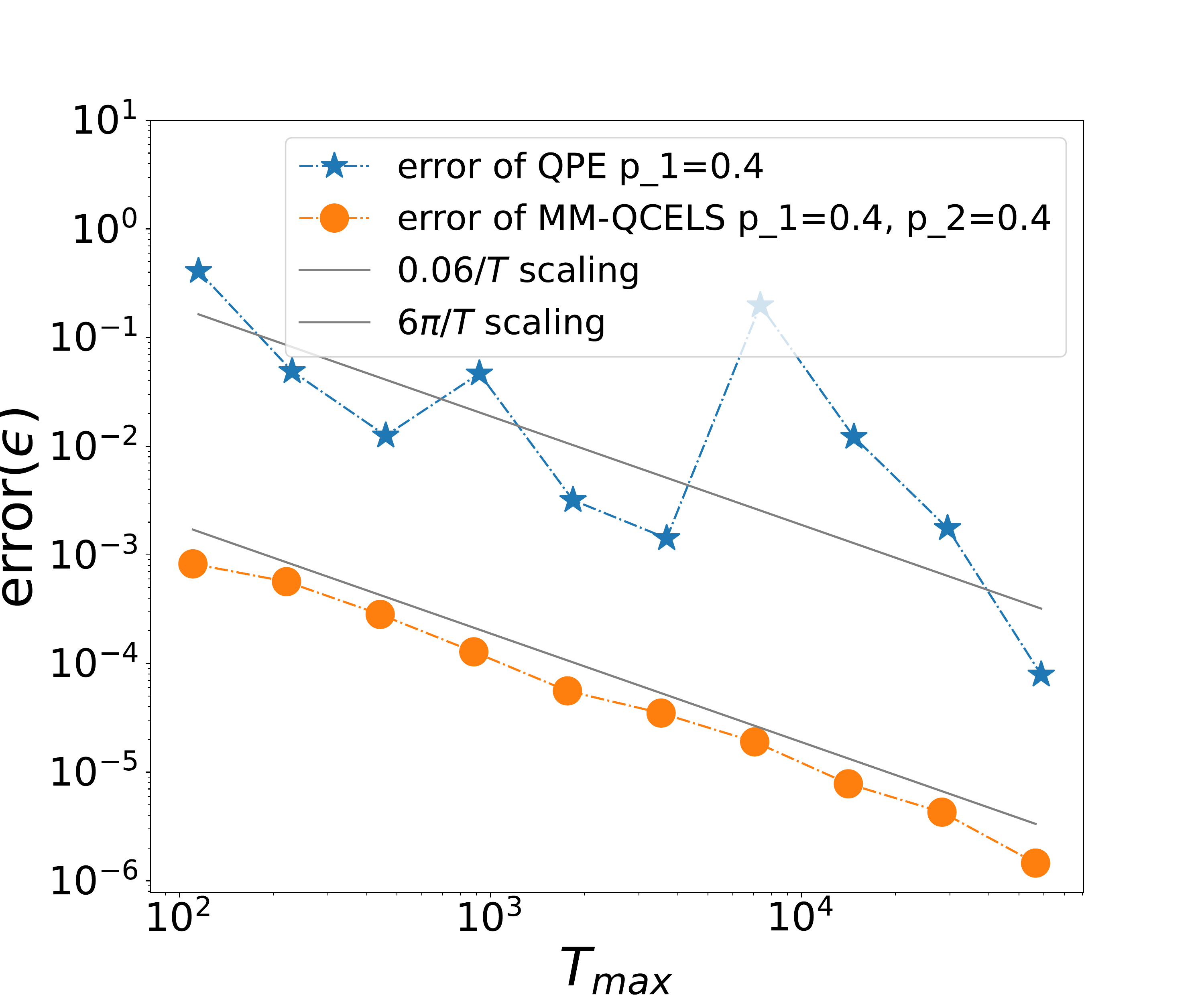}
     }
     \subfloat{
         \centering
         \includegraphics[width=0.48\textwidth]{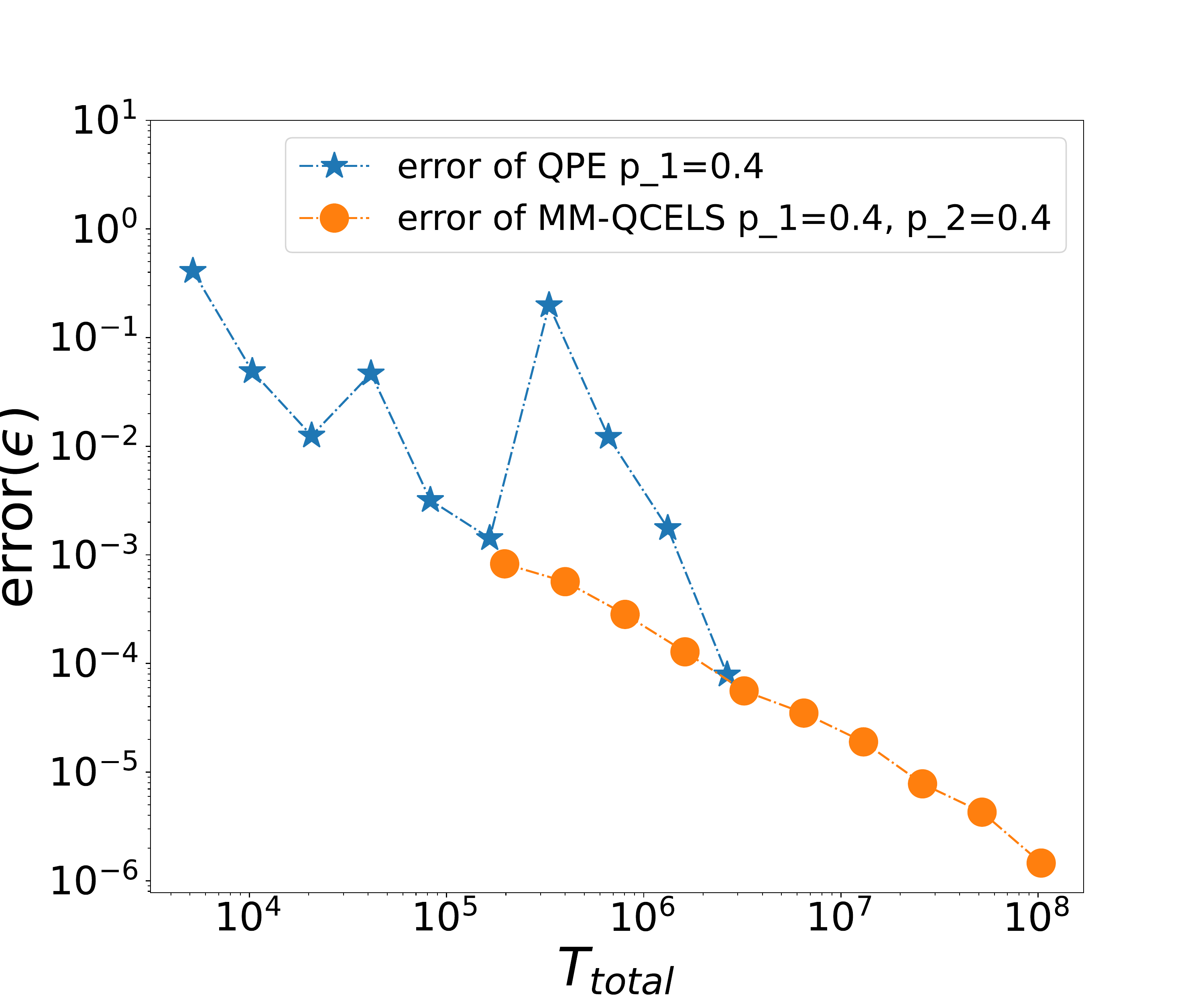}
     }
     \caption{
     \label{fig:Ising} QPE vs Algorithm \ref{alg:main}  in TFIM model with 8 sites. $p_1=p_2=0.4$. Left: Depth ($T_{\max}$); Right: Cost ($T_{\mathrm{total}}$). For Algorithm \ref{alg:main}, we choose $K=2,T_0=2(\lambda_2-\lambda_1)^{-1},N_0=3\times 10^3,N_{j\geq 1}=2\times 10^3,\gamma=1$. $l,T_j$ are chosen according to \cref{thm:gaussian}. Both methods have the error scales linearly in $1/T_{\max}$.  The constant factor $\delta=T\epsilon$ of Algorithm \ref{alg:main} is much smaller than that of QPE.}
\end{figure}

According to \cref{thm:gaussian}, accurate estimation of the dominant eigenvalues with a short circuit depth using MM-QCELS depends on two critical factors: the appropriate selection of the parameter $K$ and the fulfillment of the condition $R^{(K)}/p^{(K)}_{\min}\ll 1$. We would like to emphasize that these criteria are necessary for addressing worst-case scenarios. However, in practical implementations, even if a slightly larger value of $K$ is chosen and the ratio $R^{(K)}/p^{(K)}_{\min}$ approximates 1, Algorithm \ref{alg:main} is still possible to produce a precise approximation of the dominant eigenvalues with short circuit depth. In Appendix \ref{sec:extra_numerical}, we test these two cases and demonstrate the robustness of MM-QCELS to the choice of parameters.

\subsection{Hubbard model}\label{sec:Hubbard}
Consider the one-dimensional Hubbard model defined on $L$ spinful sites with open boundary conditions
\[
H=-t\sum^{L-1}_{j=1}\sum_{\sigma\in\{\uparrow,\downarrow\}}c^\dagger_{j,\sigma}c_{j+1,\sigma}+U\sum^L_{j=1}\left(n_{j,\uparrow}-\frac{1}{2}\right)\left(n_{j,\downarrow}-\frac{1}{2}\right).
\]
Here $c_{j,\sigma}(c^\dagger_{j,\sigma})$ denotes the fermionic annihilation (creation) operator on the site $j$ with spin $\sigma$. $\left\langle\cdot,\cdot\right\rangle$ denotes sites that are adjacent to each other. $n_{j,\sigma}=c^\dagger_{j,\sigma}c_{j,\sigma}$ is the number operator. 

We choose $L=4,8$, $t=1$, $U=10$. To implement Algorithm \ref{alg:main} and QPE, we again normalize $H$ according to \cref{eqn:normalize_H} and choose overlap $p_1=0.4,p_2=0.4$.  We run Algorithm \ref{alg:main} (with $K=2,T_0=10(\lambda_2-\lambda_1)^{-1},N_0=4\times 10^4,N_{j\geq 1}=2\times 10^3,\gamma=1$), and QPE $10$ times to compare the errors (using \cref{eqn:max_error} for Algorithm \ref{alg:main} and only estimating the error of  $\lambda_1$ for QPE). The results are shown in \cref{fig:Hubbard_4} (4 sites) and \cref{fig:Hubbard_8} (8 sites). In both cases, it can be seen that the maximal evolution time of Algorithm \ref{alg:main} is almost two orders of magnitude smaller than that of QPE. The total cost of the two methods are comparable.
%

\begin{figure}[htbp]
     \centering
     \subfloat{
         \centering
         \includegraphics[width=0.48\textwidth]{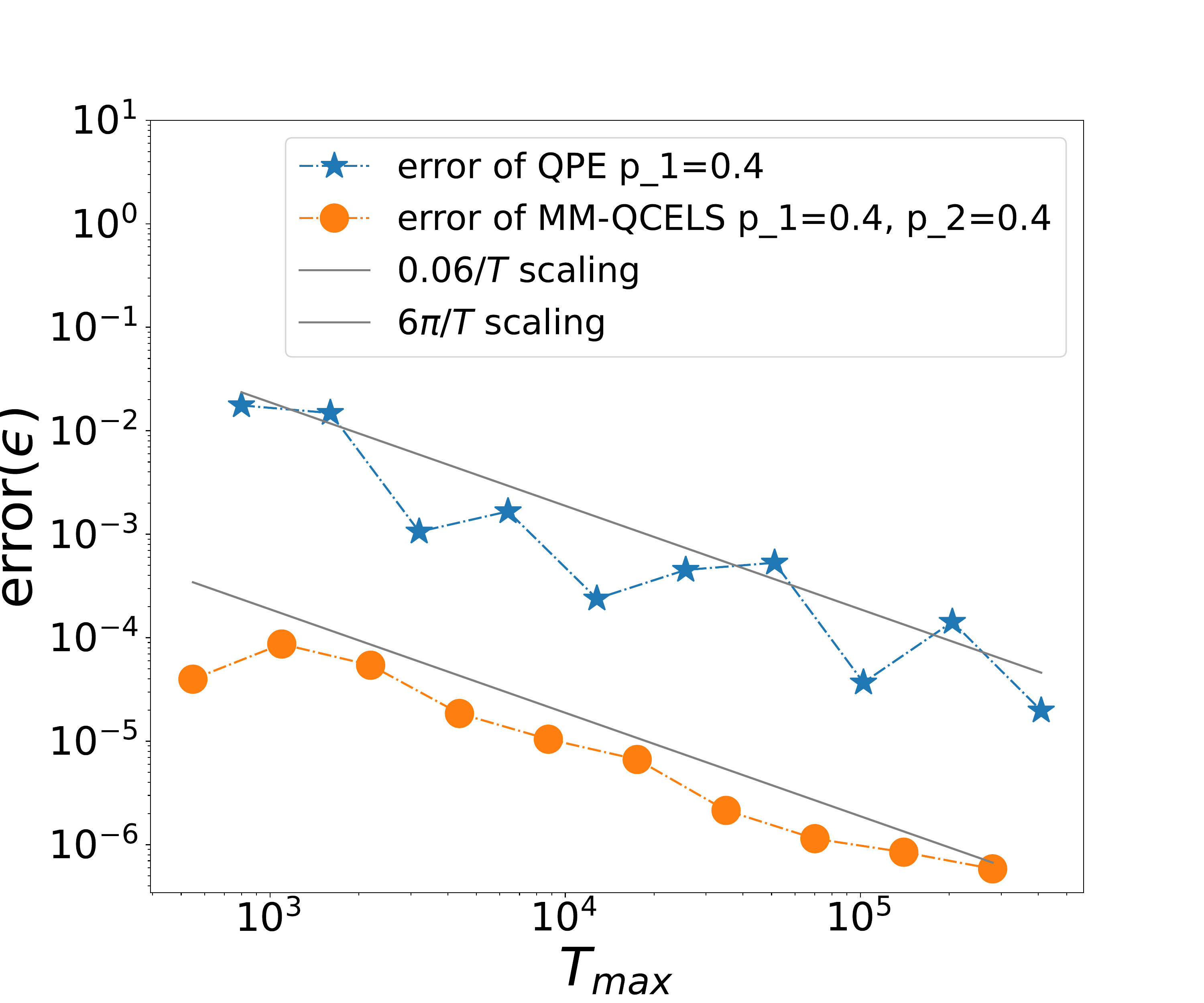}
     }
     \hfill
     \subfloat{
         \centering
         \includegraphics[width=0.48\textwidth]{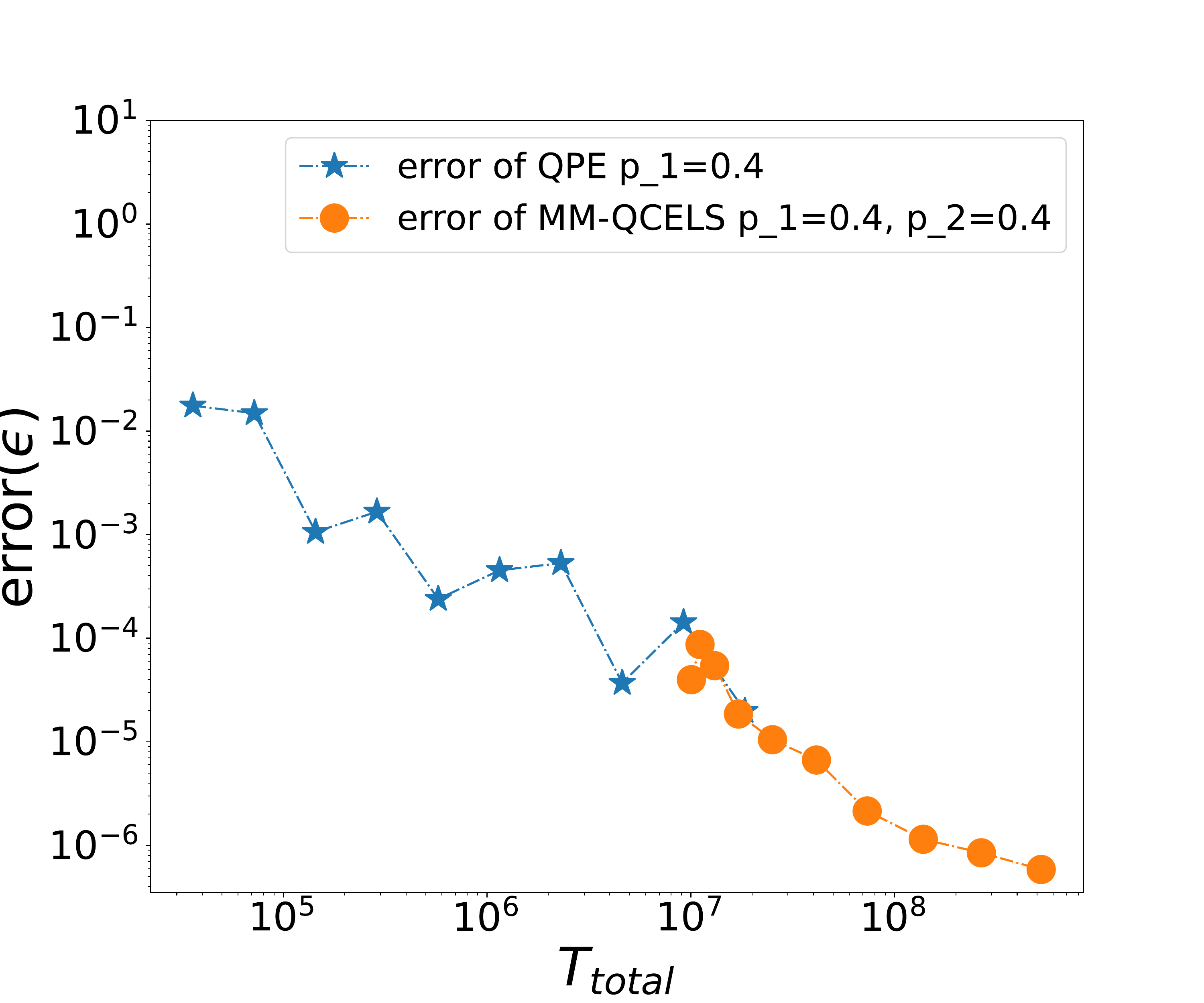}
     }
     \caption{
     \label{fig:Hubbard_4} QPE vs Algorithm \ref{alg:main} in Hubbard model with 4 sites. $p_1=p_2=0.4$. Left: Depth ($T_{\max}$); Right: Cost ($T_{\mathrm{total}}$). For Algorithm \ref{alg:main}, we choose $K=2,T_0=10(\lambda_2-\lambda_1)^{-1},N_0=4\times 10^4,N_{j\geq 1}=2\times 10^3$. $l,T_j$ are chosen according to \cref{thm:gaussian}. Compared with QPE, to achieve the same accuracy, Algorithm \ref{alg:main}  requires a much smaller circuit depth.}
\end{figure}
\begin{figure}[htbp]
     \centering
     \subfloat{
         \centering
         \includegraphics[width=0.48\textwidth]{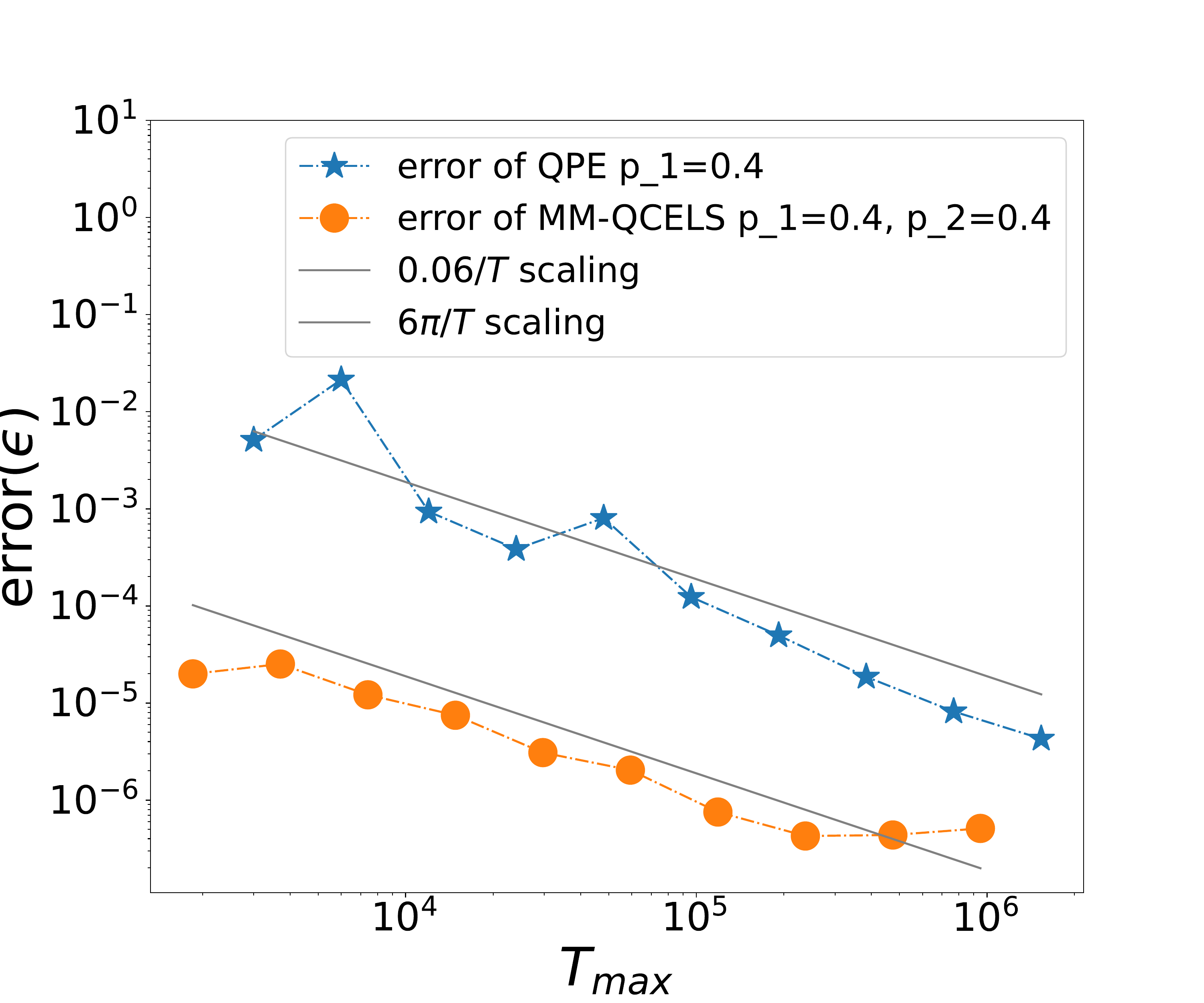}
     }
     \hfill
     \subfloat{
         \centering
         \includegraphics[width=0.48\textwidth]{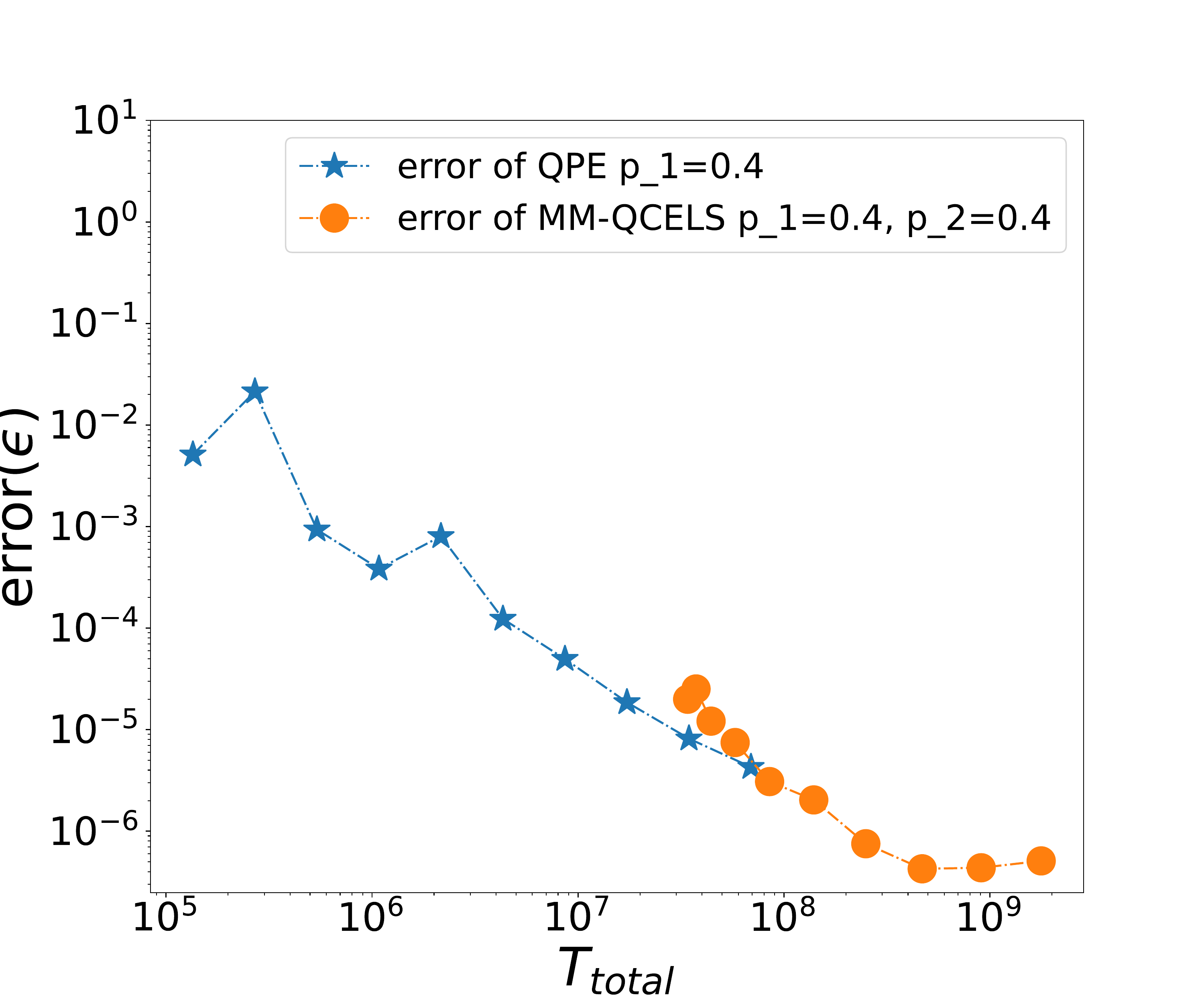}
     }
     \caption{
     \label{fig:Hubbard_8}  QPE vs Algorithm \ref{alg:main} in Hubbard model with 8 sites. $p_1=p_2=0.4$. Left: Depth ($T_{\max}$); Right: Cost ($T_{\mathrm{total}}$). For Algorithm \ref{alg:main}, we choose $K=2,T_0=10(\lambda_2-\lambda_1)^{-1},N_0=4\times 10^4,N_{j\geq 1}=2\times 10^3,\gamma=1$. $l,T_j$ are chosen according to  \cref{thm:gaussian} with $\epsilon'=0$. Compared with QPE, to achieve the same accuracy, Algorithm \ref{alg:main}   requires a much smaller circuit depth.}
\end{figure}

\FloatBarrier
\section{Discussion}\label{sec:discussion}

In this paper, we have developed a new algorithm to simultaneously estimate multiple eigenvalues using a simple circuit. This optimization-based approach, called multi-modal, multi-level quantum complex exponential least squares (MM-QCELS) saturates the Heisenberg-limited scaling, and achieves a relatively short maximal running time when the residual overlap (denoted by $R^{(K)}$) is small. With a proper initial state, this algorithm can be used to efficiently estimate ground-state and excited-state energies of a quantum many-body Hamiltonian on early fault-tolerant quantum computers. 

There are a number of directions to extend this work and to strengthen the analysis. 

\begin{enumerate}

\item If the initial state has significant high energy contributions, it can be  implicitly filtered using the circuit in \cref{fig:qc}~\cite{LinTong2022}, or explicitly filtered using quantum eigenvalue transformation of unitary matrices (QETU) \cite{dong2022ground} to satisfy the condition of small residual overlap. Similar to the discussion in Ref.~\cite{DingLin2023}, this does not necessarily require a large spectral gap  between the dominant eigenvalues and the non-dominant ones, and a \textit{relative overlap} condition (which is a property of the initial state rather than the Hamiltonian) could suffice.

\item Our complexity analysis depends the minimal dominant spectral gap $\Delta^{(K)}_{\lambda}$. This is a necessary condition, since the simulation time should be long enough in order to separate the eigenvalues. If two or more eigenvalues are nearly degenerate (i.e., the distance is less than $\epsilon$), we may combine them and view these nearly degenerate eigenvalues as a single eigenvalue. The MM-QCELS method can still be applied without changes. However, compared to the result in \cref{thm:gaussian}, there may not be a one-to-one correspondence between the estimated eigenvalues $\theta_k$ and the dominant eigenvalues $\lambda_{m_k}$.


\item Due to the complex energy landscape of the loss function, if we use a grid search to find the global minima, the cost of solving the optimization problem in \cref{eqn:op} on a \textit{classical computer} may scale exponentially in $K$ in the worst case. On the other hand, the signal processing method in \cite{price_2015,ni2023lowdepth_2} can scale polynomially in $K$, but the implementation can be much more complicated than ours. Therefore it is desirable to improve our algorithm (e.g., using a robust initialization strategy) to reduce the cost on the classical post-processing step for large $K$.

\item While Algorithm \ref{alg:main} is formulated as a multi-level optimization problem, as discussed in Remark \ref{re:add}, solving \eqref{eqn:op} with sufficiently large values of $T$ and $N$ once may be enough to approximate the dominant eigenvalues. This approach differs from the signal processing-based method discussed in~\cite{ni2023lowdepth_2}, which requires multi-level signal processing.

\end{enumerate}

\vspace{1em}
\textbf{Acknowledgment}

This material is based upon work supported by the U.S. Department of Energy, Office of Science, National Quantum Information Science Research Centers, Quantum Systems Accelerator (L.L.). Additional support is acknowledged from
the NSF Quantum Leap Challenge Institute (QLCI)
program under Grant number OMA-2016245 (Z.D.), and a Google Quantum Research
Award (L. L.). L.L. is a Simons Investigator. The authors thank Haoya Li, Hongkang Ni, Lexing Ying, Ruizhe Zhang for helpful discussions.

\bibliographystyle{abbrvurl}
\bibliography{ref}
\clearpage
\appendix

The appendix is organized as follows:
\\
\begin{enumerate}[leftmargin=.5cm,itemindent=.5cm,labelwidth=\itemindent,labelsep=0cm,align=left]
    \item[Appendix \ref{sec:extra_numerical}.] We give two extra numerical tests to demonstrate the robustness of MM-QCELS in relation to parameter selection.
    \item[Appendix \ref{sec:pf_intuitive}.] We give an intuitive proof for \cref{thm:gaussian} in this section.
    \item[Appendix \ref{sec:proof_of_thm}.] We give a rigorous proof for \cref{thm:gaussian} in this section.
    \item[Appendix \ref{sec:pf_prop}.] We give the proof of  \cref{prop:single_op_problem}, which is a key in the proof of \cref{thm:gaussian}. 
    \item[Appendix \ref{sec:bound_expectation_error}.] We provide some other technical estimations  used for our proof.
\end{enumerate}

\section{Additional numerical experiments}\label{sec:extra_numerical}
In this section, we numerically demonstrate the robustness of \cref{alg:main} using the one-dimensional transverse field Ising model (TFIM) model with $L=8$, $g=4$, as defined in \cref{sec:Ising}. We also normalize the spectrum of the original Hamiltonian $H$ using \eqref{eqn:normalize_H}.

In our first test, we set $p_1=0.7$, $p_2=0.2$, and $p_k=1/2540$ for $3\leq k\leq 256$ (so that $\sum_k p_k=1$). Therefore according to \cref{thm:gaussian}, we should choose $K=1$ or $K=2$. We apply Algorithm \ref{alg:main} (with $K=2,3,4$\footnote{The output of the algorithm comprises $K$ pairs of approximated eigenvalues and corresponding weights denoted as ${(r_i,\theta^\star_i)}^K_{i=1}$. From this set, we identify the two approximated eigenvalues with the greatest absolute weights, considering them to be the approximated dominant eigenvalues.},$T_0=10(\lambda_2-\lambda_1)^{-1},N_0=3\times 10^3,N_{j\geq 1}=2\times 10^3$ and $\gamma=1$) to estimate $(\lambda_1,\lambda_2)$ and compute the maximal error~\eqref{eqn:max_error}. We also run QPE $10$ times \textit{only} to estimate $\lambda_1$. The comparison of the results is shown in \cref{fig:Ising_2}. Surprisingly, although $R^{(K)}/p^{(K)}_{\min}\gg 1$ when we choose $K=3,4$, meaning that it does not satisfy the condition of \cref{thm:gaussian}, Algorithm \ref{alg:main} still estimates $(\lambda_1.\lambda_2)$ accurately with short circuit depth $(T_{\max})$. Moreover, the total evolution time ($T_{\mathrm{total}})$ of Algorithm \ref{alg:main} is similar to that of QPE. One intuitive explanation for this phenomenon can be derived from the proof of \cref{thm:gaussian} in Appendix \ref{sec:pf_intuitive}. When $K>2$, aiming to reach the global minimum leads us to anticipate the existence of a unique pair $(k_1, k_2) \in \{1, 2, \cdots, K\}^{\otimes 2}$ with the property that $\theta^*_{k_1} \approx \lambda_1$ and $\theta^*_{k_2} \approx \lambda_2$. Then, analogous to the derivation of \eqref{eqn:depth_improve}, we can show $|\theta^*_{k_1}-\lambda_{1}|= \mathcal{O}\left(\frac{R^{(K)}}{p_{1}T}\right)$ and $T|\theta^*_{k_2}-\lambda_{2}|= \mathcal{O}\left(\frac{R^{(K)}}{p_{2}T}\right)$. These equalities imply that, even if a wrong choice of $K$ is chosen, as long as $R^{(K)}/p_{1} \ll 1$ and $R^{(K)}/p_{2} \ll 1$ hold true, an accurate estimation of the dominant eigenvalues can still be obtained with relatively short circuit depth.

The second test in this section focuses on studying the effect of small $p^{(K)}_{\min}$ (or $R^{(K)}/p^{(K)}_{\min}\approx 1$). We construct the initial state with $p_1=0.21,p_2=0.6$ and set $K=2$. In this setting, we have $R^{(K)}/p^{(K)}_{\min}=0.19/0.21\approx 1$. We apply Algorithm \ref{alg:main} (with $K=2$,$T_0=10(\lambda_2-\lambda_1)^{-1},N_0=3\times 10^3,N_{j\geq 1}=2\times 10^3$ and $\gamma=1$) to estimate $(\lambda_1,\lambda_2)$ and compute the maximal error~\eqref{eqn:max_error}. We also run QPE $10$ times to estimate $\lambda_1$ \textit{only}. The results is summarized in \cref{fig:Ising_3}. Interestingly, although \cref{thm:gaussian} requires $R^{(K)}/p^{(K)}_{\min}\ll 1$ to achieve short circuit depth, MM-QCELS still demonstrates superior performance in circuit depth compared to QPE in this case. Intuitively, this phenomenon finds its explanation through reasoning analogous to the first point raised in Remark \ref{rem:after_intuition}. To be more specific, for any given $\delta>0$ and a small enough $\epsilon$, when $T\geq \delta/\epsilon\gg \Delta^{-1}_\lambda$, where $\Delta_\lambda$ is the spectral gap between the dominant eigenvalues and the non-dominant ones, we can establish that $|\theta^*_k-\lambda_{m_k}|\leq \epsilon$. Here, $\theta^*_k$ is the solution to the ideal loss function~\eqref{eqn:loss_multi_modal_perfect}. Furthermore, we can select the number of data points as $N=\widetilde{\Theta}(\delta^{-2-o(1)})$ to effectively control the random measurement error. This approach ultimately enables us to attain Heisenberg-limited scaling and a short circuit depth.

\begin{figure}[htbp]
     \subfloat{
         \centering
         \includegraphics[width=0.48\textwidth]{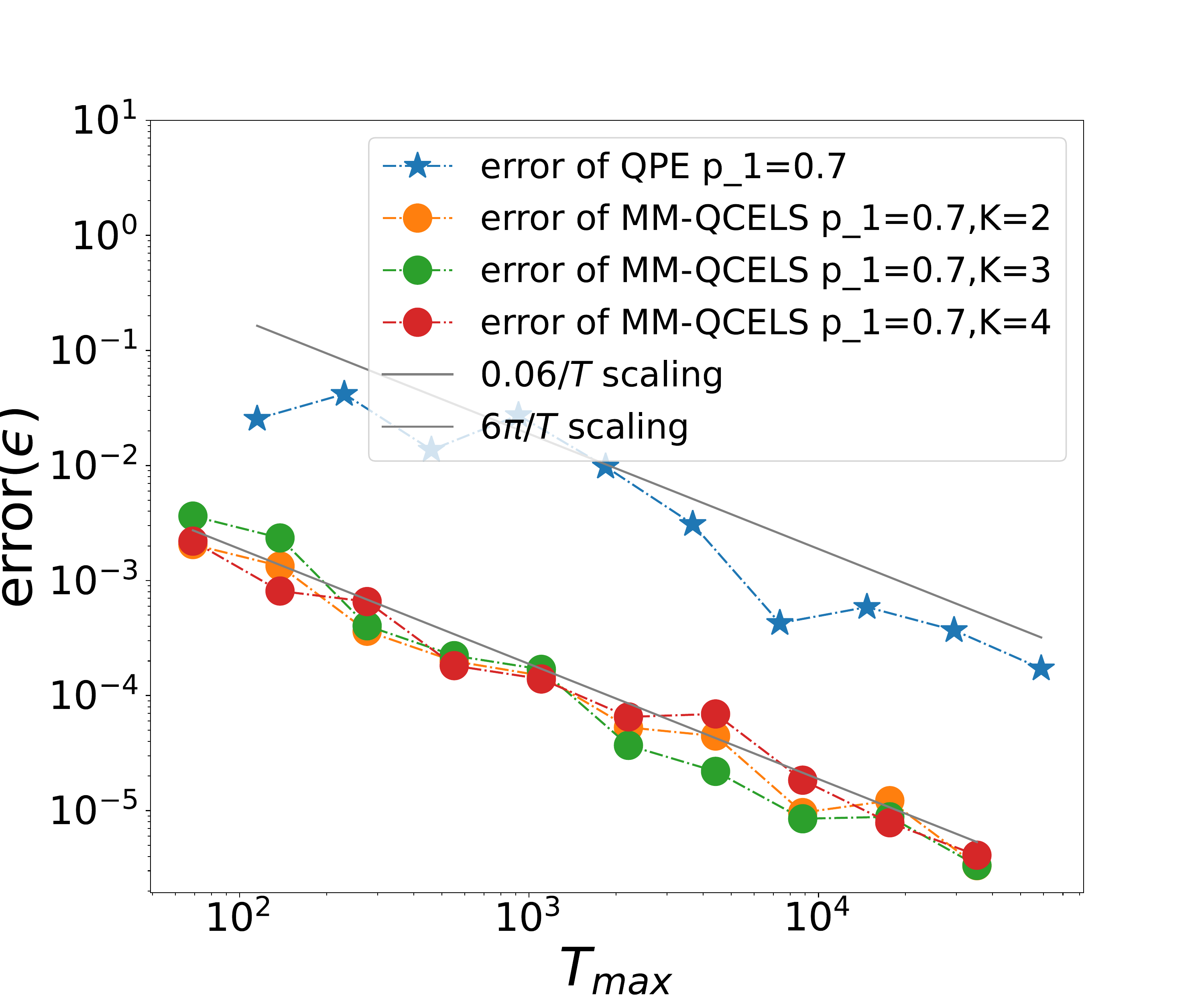}
     }
     \subfloat{
         \centering
         \includegraphics[width=0.48\textwidth]{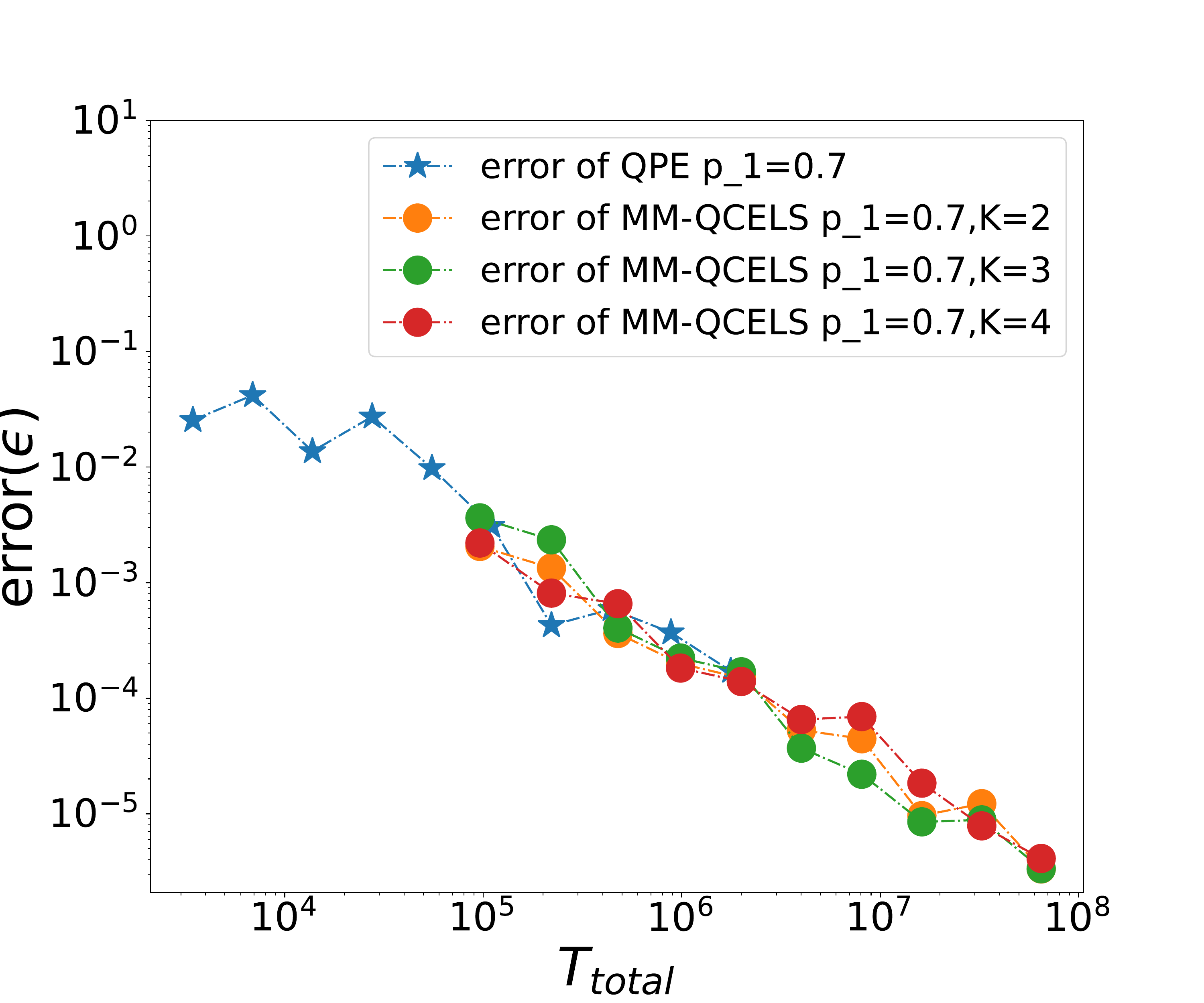}
     }
     \caption{
     \label{fig:Ising_2} QPE vs Algorithm \ref{alg:main}  in TFIM model with 8 sites. $p_1=0,7,p_2=0.2,p_k=1/2540$ for $k\geq 3$. Left: Depth ($T_{\max}$); Right: Cost ($T_{\mathrm{total}}$). For Algorithm \ref{alg:main}, we choose $K=2,3,4,T_0=10(\lambda_2-\lambda_1)^{-1},N_0=3\times 10^3,N_{j\geq 1}=2\times 10^3,\gamma=1$. $l,T_j$ are chosen according to \cref{thm:gaussian}. Both methods have the error scales linearly in $1/T_{\max}$.  The constant factor $\delta=T\epsilon$ of Algorithm \ref{alg:main} is much smaller than that of QPE.}
\end{figure}

\begin{figure}[htbp]
     \subfloat{
         \centering
         \includegraphics[width=0.48\textwidth]{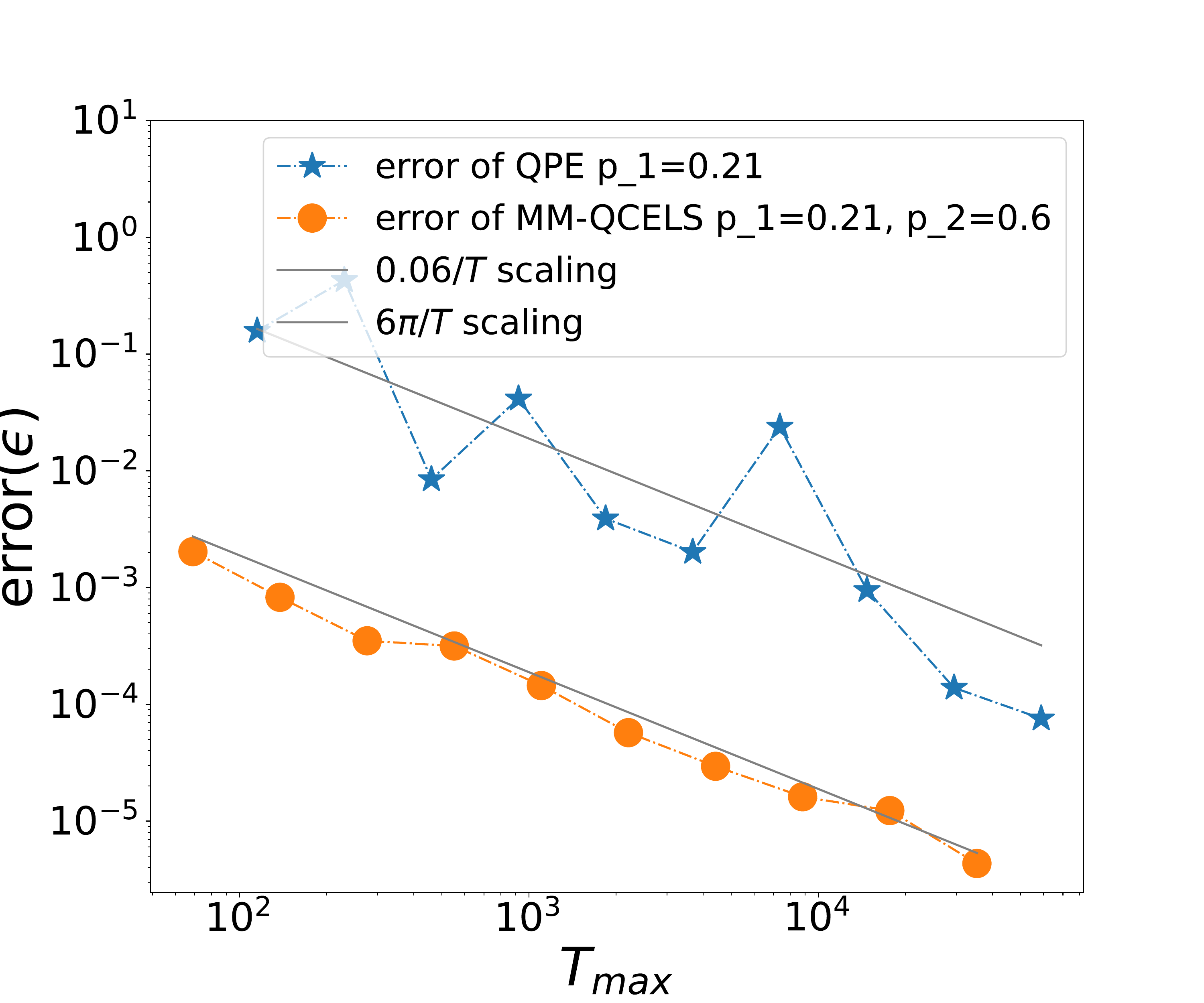}
     }
     \subfloat{
         \centering
         \includegraphics[width=0.48\textwidth]{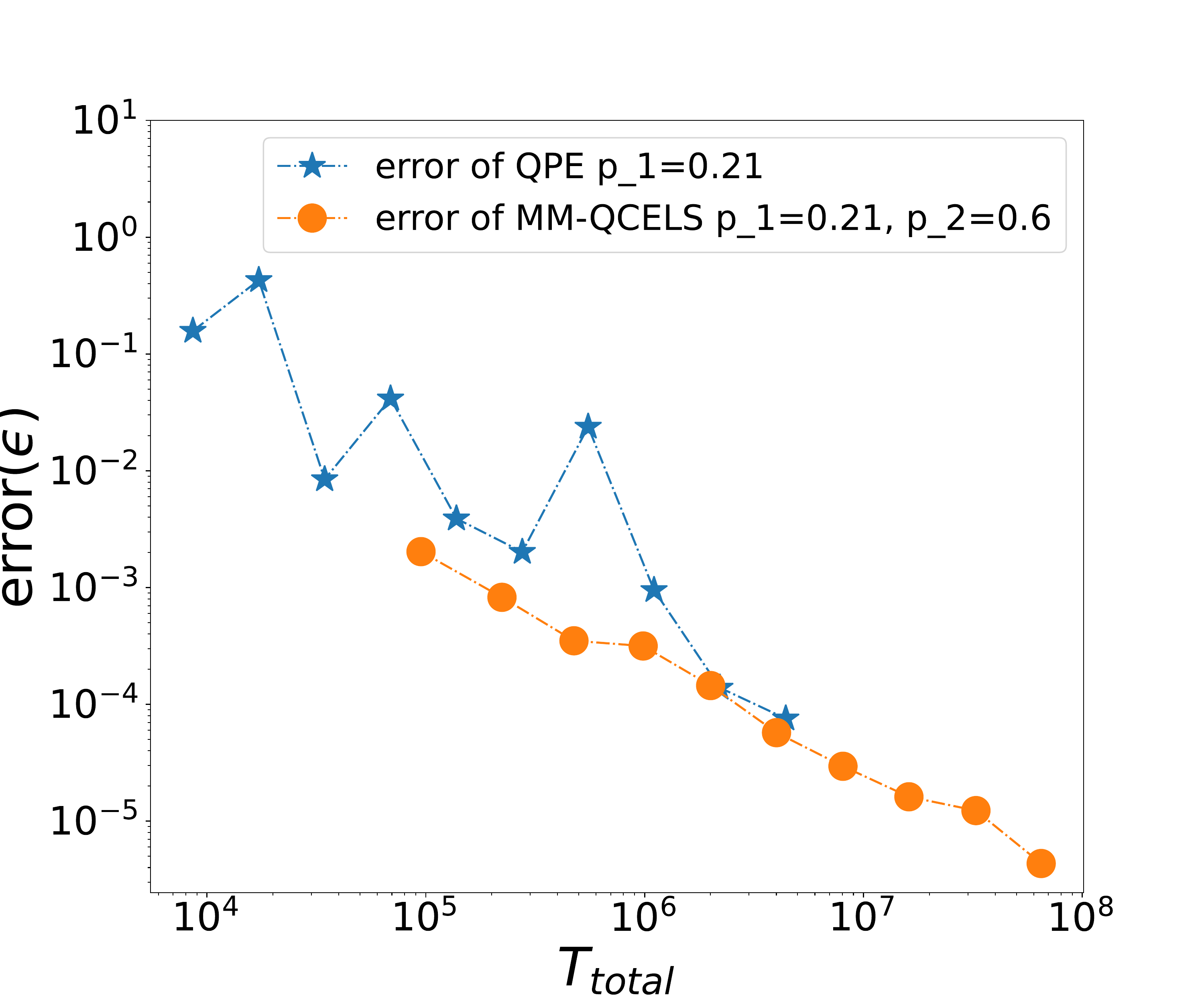}
     }
     \caption{
     \label{fig:Ising_3} QPE vs Algorithm \ref{alg:main}  in TFIM model with 8 sites. $p_1=0,21,p_2=0.6$. Left: Depth ($T_{\max}$); Right: Cost ($T_{\mathrm{total}}$). For Algorithm \ref{alg:main}, we choose $K=2,T_0=10(\lambda_2-\lambda_1)^{-1},N_0=3\times 10^3,N_{j\geq 1}=2\times 10^3,\gamma=1$. $l,T_j$ are chosen according to \cref{thm:gaussian}. Both methods have the error scales linearly in $1/T_{\max}$.  The constant factor $\delta=T\epsilon$ of Algorithm \ref{alg:main} is much smaller than that of QPE.}
\end{figure}

\section{An intuitive proof of Theorem \ref{thm:gaussian}}\label{sec:pf_intuitive}

In this section, we first give an informal derivation to show that by solving the optimization problem, it is possible to find an accurate approximation to dominant eigenvalues with a short maximal running time. For simplicity, in the inituitive proof, we only consider the case when
\begin{enumerate}
\item The minimal dominant spectral gap $\Delta^{(K)}_\lambda$ is much larger than the precision $\epsilon$: $\Delta^{(K)}_\lambda\gg \epsilon$;
\item The modes in $\mc{D}$ are dominant: $p^{(K)}_{\min}>CR^{(K)}$ for some constant $C>0$;
\item The maximal runtime is sufficient for resolving the dominant eigenvalues: $T_{\max}\ge C'/\Delta^{(K)}_\lambda$ for some constant $C'$>0.
\end{enumerate}
Here, $p^{(K)}_{\min}$ is defined in \eqref{eqn:pK_min}, $\Delta^{(K)}_{\lambda}$ is defined in \eqref{eqn:DeltaK_lambda}, and $R^{(K)}$ is defined in \eqref{eqn:RK}.

For now we only focus on the ideal loss function, which can be rewritten as 
\begin{equation}\label{eqn:L_K_ideal}
\mathcal{L}_K\left(\vr,\vtheta\right)=\vr^\dagger U\left(\vtheta\right) \vr-\left(V^\dagger\left(\vtheta\right)\vr+\vr^\dagger V\left(\vtheta\right)\right)+W.
\end{equation}
Here $U\left(\vtheta\right)\in\mathbb{C}^{K\times K}$, $V\left(\vtheta\right)\in\mathbb{C}^{(K)}$, and $W\in\mathbb{R}$ are defined as
\[
U_{k_1,k_2}=F(\theta_{k_1}-\theta_{k_2}),\ V_{k}=\sum^M_{m=1}p_mF(\theta_k-\lambda_m),\  W=\int^\infty_{-\infty}a(t)\left|\sum^M_{m=1} p_{m}\exp\left(-i \lambda_{m} t\right)\right|^2\rd t 
\]
for $1\leq k,k_1,k_2\leq K$. For simplicity, in this intuitive analysis we also neglect the difference between the truncated Gaussian distribution and the Gaussian distribution, i.e., 
\begin{equation}
F(x)=\int^\infty_{-\infty}a(t)\exp(ixt)\rd t\approx \exp(-T^2x^2/2)\,.
\end{equation}

Denote $\mathcal{D}=\{m_1,m_2,\dots,m_K\}$, where $m_1<m_2<\dots<m_K$. Without loss of generality, assume the minimizer satisfies $\theta^*_1\leq \theta^*_2\leq \dots\leq \theta^*_K$.
We first claim (without giving the proof here) that when $T=\Omega\left(1/\Delta^{(K)}_\lambda\right)$,
\begin{equation}\label{eqn:intuitive_1}
\left|\lambda_{m_k}-\theta^*_{k}\right|\leq \frac{\Delta^{(K)}_{\lambda}}{4},\quad \forall 1\leq k\leq K\,.
\end{equation}
In other words, each $\theta_k^*$ approximates a unique dominant eigenvalue up to a constant proportional to the minimal dominant spectral gap.  The next step is to refine the eigenvalue estimate to the target precision $\epsilon$.

When $T=\Omega\left(1/\Delta^{(K)}_\lambda\right)$, the matrix $U$ is approximately the identity matrix, and $V_k\approx p_{m_k}F(\theta_k-\lambda_{m_k})$. This gives
\begin{equation}\label{eqn:approximation_1}
\mathcal{L}_K\left(\vr,\vtheta\right)\approx \sum^{K}_{k=1}\left(|r_k|^2-2\mathrm{Re}(p_{m_k}r_k)F(\theta_k-\lambda_{m_k})-2\sum_{m'\in\mathcal{D}^c}\mathrm{Re}(p_{m'}r_k)F(\theta_k-\lambda_{m'})\right)+W\,.
\end{equation}
Hence conceptually, we can solve for each pair $(r_k,\theta_k), 1\leq k\leq K$ \textit{independently} as 
\begin{equation}\label{eqn:indenpent_r_k_theta_k}
(r^*_k,\theta^*_k)=\argmin_{r,\theta}|r|^2-2\mathrm{Re}(p_{m_k}r)F(\theta-\lambda_{m_k})-2\sum_{m'\in\mathcal{D}^c}\mathrm{Re}(p_{m'}r)F(\theta-\lambda_{m'})\,.
\end{equation}
Consider the minimization problem on the right-hand side with fixed $k$. Noticing that this new loss function is quadratic in $r$, we obtain that
\begin{equation}\label{eqn:theta_k_intuitive}
    \theta^*_k=\argmax_{\theta}p_{m_k}F(\theta-\lambda_{m_k})+\sum_{m'\in\mathcal{D}^c}p_{m'}F(\theta-\lambda_{m'})
\end{equation}
Plugging $\theta=\lambda_{m_k}$ in the right-hand side, we obtain
\begin{equation}\label{eqn:theta_k_intuitive_compare}
p_{m_k}F(\theta^*_k-\lambda_{m_k})+\sum_{m'\in\mathcal{D}^c}p_{m'}F(\theta^*_k-\lambda_{m'})\geq p_{m_k}+\sum_{m'\in\mathcal{D}^c}p_{m'}F(\lambda_{m_k}-\lambda_{m'})\,.
\end{equation}
Using the Gaussian approximation $F(x)=\exp(-T^2x^2/2)$, we have \[\left|F'(x)\right|=T^2\left|x\right|\exp(-T^2x^2/2)\leq \left(\sup_{z\in\mathbb{R}} z\exp(-z^2/2)\right)T=\Theta(T)\,,\]
where we view $T|x|$ as $z$ in the inequality. This implies $F(x)$ is a $\mathcal{O}(T)$-Lipschitz function. Combining this with \eqref{eqn:theta_k_intuitive_compare}, we obtain
\begin{equation}\label{eqn:exponential_lowerbound}
\begin{aligned}
&\exp\left(-\frac{T^2(\theta^*_k-\lambda_{m_k})^2}{2}\right)-1\\
\geq &\sum_{m'\in\mathcal{D}^c}\frac{p_{m'}}{p_{m_k}}\left(F(\lambda_{m_k}-\lambda_{m'})-F(\theta^*_k-\lambda_{m'})\right)\\
\geq&-\frac{R^{(K)}}{p_{m_k}}\min\{\mathcal{O}(|T(\theta^*_k-\lambda_{m_k})|),1\}.
\end{aligned}
\end{equation}
where we use $0\leq F\leq 1$  and $F$ is a $\mathcal{O}(T)$-Lipschitz function in the last inequality. From \eqref{eqn:exponential_lowerbound}, we first have $\exp\left(-\frac{T^2(\theta^*_k-\lambda_{m_k})^2}{2}\right)\geq 1-\frac{R^{(K)}}{p_{m_k}}$, which implies
\begin{equation}
T|\theta^*_k-\lambda_{m_k}|=\mathcal{O}\left(\sqrt{\frac{R^{(K)}}{p_{m_k}}}\right)\,.
\label{eqn:qcels_estimate}
\end{equation}

When $R^{(K)}$ is sufficiently small, combining \cref{eqn:exponential_lowerbound,eqn:qcels_estimate}, we further obtain 
\begin{equation}\label{lower_bound_exp_intuitive}
1-\frac{T^2(\theta^*_k-\lambda_{m_k})}{4}\ge \exp\left(-\frac{T^2(\theta^*_k-\lambda_{m_k})^2}{2}\right)\geq 1-\frac{R^{(K)}}{p_{m_k}}\mathcal{O}(|T(\theta^*_k-\lambda_{m_k})|)\,.
\end{equation}
This implies 
\begin{equation}
T|\theta^*_k-\lambda_{m_k}|= \mathcal{O}\left(\frac{R^{(K)}}{p_{m_k}}\right).
\label{eqn:depth_improve}
\end{equation}

In summary, to obtain $\left|\theta^*_k-\lambda_{m_k}\right|\leq \epsilon$ for all $1\leq k\leq K$, we can set
\begin{equation}\label{eqn:rough_informal_short_depth}
    T=\Theta\left(\frac{R^{(K)}}{p^{(K)}_{\min}\epsilon}\right)\,.
\end{equation}
This implies the depth constant of maximal running time $T_{\max}=\Theta\left(\frac{ R^{(K)}}{p^{(K)}_{\min}\epsilon}\right)$ is much smaller than 
$\frac{\pi}{\epsilon}$ when $R^{(K)}/p^{(K)}_{\min}$ is close to $0$.

We remark that when $K=1$, there is only one dominant mode $m_k$, and $p_{m_k}=1-R^{(K)}$ by definition. In this case, the result in \cref{eqn:qcels_estimate} is comparable to the estimate in Ref.~\cite{DingLin2023} for the QCELS method. The analysis in this work provides a tighter bound of the maximal runtime (or the circuit depth). Specifically, \cref{eqn:depth_improve} provides a quadratic improvement with respect to the preconstant $R^{(K)}$ for estimating a single dominant eigenvalue, and the same conclusion holds for estimating multiple eigenvalues.

\begin{rem}\label{rem:after_intuition}
\begin{enumerate}

\item When a spectral gap exists between the dominant eigenvalues and the non-dominant ones (represented as $\Delta_{\lambda}$), we can further reduce the maximum runtime to $T_{\max}=\wt{\Theta}(1/\min\{\Delta_{\lambda},\Delta^{(K)}_{\lambda}\})$.
This reduction can be derived in a similar manner to the previous intuitive analysis. Specifically, referring to equation \eqref{eqn:exponential_lowerbound}, when $T$ is sufficiently large to ensure $\left|\theta^{(K)}-\lambda_{m_k}\right|\leq \frac{\Delta_{\lambda}}{2}$, we have \[\left|F(\lambda_{m_k}-\lambda_{m'})-F(\theta^{(K)}-\lambda{m'})\right|=\mathcal{O}(\exp(-T^2\Delta_{\lambda}^2/8))\,.\] Following a similar derivation to \eqref{eqn:qcels_estimate}-\eqref{eqn:depth_improve}, we obtain the following expression:
\[
T|\theta^*_k-\lambda_{m_k}|= \mathcal{O}\left(\frac{\exp(-T^2\Delta_{\lambda}^2/8)}{p_{m_k}}\right)\,,
\]
This implies that $T_{\max}=\wt{\Theta}(1/\min\{\Delta_{\lambda},\Delta^{(K)}_{\lambda}\})$ is sufficient to ensure $\epsilon$-accuracy. However, in this scenario, while $T_{\max}$ logarithmically depends on the desired precision $\epsilon$, the number of data points needs to increase to $N=\Omega(\epsilon^{-2})$ to adequately control the random noise. This is similar in flavor to the result in \cite{Wang_2022} as well as in QCELS \cite{DingLin2023} for estimating a single dominant eigenvalue.

\item The rigorous proof of Theorem \ref{thm:gaussian} follows a slightly different path from the previous intuitive derivation. The numerical loss function \eqref{eqn:loss_multi_modal} admits a similar quadratic expansion as in \cref{eqn:L_K_ideal} with noisy coefficients $U(\theta)$, $V(\theta)$, and $W$. Due to the presence of noise and off-diagonal entries in $U$, the perfect separation assumed in \eqref{eqn:indenpent_r_k_theta_k} no longer holds, and the analysis of the independent optimization problem cannot be directly applied to show that $\theta_k^*$ is close to $\lambda_{m_k}$. To overcome this difficulty, we adopt the idea of separation and decompose the numerical loss function after bounding the noise. By comparing the resulting loss function with $L_K\left(\{p_m\}_{m\in\mathcal{D}},\{\lambda_m\}_{m\in\mathcal{D}}\right)$,  we demonstrate that $\exp\left(-\frac{T^2(\theta^*_k-\lambda_{m_k})^2}{2}\right)-1= \mathcal{O}(q)$ and $|\theta_k^*-\lambda_{m_k}|=\mathcal{O}(q/T)$, where $q=\Theta(R^{(K)}/p_{\min}^{(K)})$. This implies, to obtain $\epsilon$-accuracy, the maximal running time $T_{\max}=\gamma T=\delta/\epsilon$, where $\delta=\widetilde{\Theta}\left(q\log(q^{-1})\right)$.

\end{enumerate}
\end{rem}

\section{Rigorous proof of Theorem \ref{thm:gaussian}}\label{sec:proof_of_thm}

To prove Theorem \ref{thm:gaussian}, we first rewrite the optimization problem \eqref{eqn:op} as 
\begin{equation}\label{eqn:separation_of_error}
\begin{aligned}
&\argmin_{\vr,\vtheta}L_{K}\left(\vr,\vtheta\right)\\
=&\argmin_{\vr,\vtheta}\frac{1}{N}\sum^N_{n=1}\left|\sum^{M}_{m=1}p_m\exp(-i \lambda_m t_n)-\sum^{K}_{k=1}r_k\exp(-i\theta_k t_n)\right|^2\\
&-\frac{2}{N}\sum^N_{n=1}\mathrm{Re}\left(\left\langle E_n,\sum^{K}_{k=1}r_k\exp(-i\theta_k t_n)\right\rangle\right)\\
= &\argmin_{\vr,\vtheta}\int^\infty_{-\infty}a(t)\left|\sum^{M}_{m=1}p_m\exp(-i \lambda_m t)-\sum^{K}_{k=1}r_k\exp(-i\theta_k t)\right|^2\rd t+E_{p,r}+E_{r,r}+E_{r,Z}\\
\end{aligned}
\end{equation}
where we omit term $\frac{1}{N}\sum^N_{n=1}|
E_n|^2$ in the first equality and $p_mp_{m'}\exp(i(\lambda_m-\lambda_{m'})t)$ terms in the second equality because they do not affect the solution of the optimization problem. Here
\begin{equation}\label{eqn:formula_E}
\begin{aligned}
&E_{p,r}(\vr,\vtheta)=2\mathrm{Re}\left(\sum^{M}_{m=1}\sum^{K}_{k=1}p_mr_k\left(\frac{1}{N}\sum^N_{n=1}\exp(i(\theta_k- \lambda_m) t_n)-F(\theta_k-\lambda_m)\right)\right)\,,\\
&E_{r,r}(\vr,\vtheta)=2\mathrm{Re}\left(\sum^{K}_{k_1\neq k_2}\overline{r}_{k_1}r_{k_2}\left(\frac{1}{N}\sum^N_{n=1}\exp(i(\theta_{k_2}-\theta_{k_1}) t_n)-F(\theta_{k_2}-\theta_{k_1})\right)\right)\,,\\
&E_{r,Z}(\vr,\vtheta)=-2\mathrm{Re}\left(\sum^{K}_{k=1}\overline{r}_k\left(\frac{1}{N}\sum^N_{n=1}E_n\exp(i\theta_k t_n)\right)\right)\,.
\end{aligned}
\end{equation}
Roughly speaking, when $N\gg 1$, we have the expectation error $|E_{p,r}|,|E_{r,r}|,|E_{r,Z}|= \mathcal{O}(1/\sqrt{N})$ (see Appendix \ref{sec:bound_expectation_error} for detail). 

The proof contains two steps: 1. Find initial estimation intervals for the dominant eigenvalues (using $T_0,N_0$); 2. In the correct estimation intervals, find more accurate approximations of the dominant eigenvalues (from $T_{j}$ to $T_{j+1}$). Now, we introduce a lemma and a proposition to control the complexity of these two steps respectively.

Let $E(\vr,\vtheta)=E_{p,r}(\vr,\vtheta)+E_{r,r}(\vr,\vtheta)+E_{r,Z}(\vr,\vtheta)$. We use the following lemma to control the complexity of the first step:
\begin{lem}\label{lem:single_op_problem}  
Given $0<q<1$ such that $q=\Omega(R^{(K)}/p^{(K)}_{\min})$, where $p^{(K)}_{\min}$ and $R^{(K)}$ are defined in \eqref{eqn:pK_min} and \eqref{eqn:RK}, we assume $p^{(K)}_{\min}>3R^{(K)}$. Define
  \[
  \left(\vr^*,\vtheta^*\right)=\argmin_{\|\vr\|_1\leq 1,\theta_k\in [-\pi,\pi]}L_{K}\left(\vr,\vtheta\right)\,,
  \]
where $L_{K}$ is defined in \eqref{eqn:loss_multi_modal}. If
\begin{equation}\label{eqn:condition_gaussian_lemma}
\begin{aligned}
  &\gamma =\Theta\left(\log\left(1/\min\left\{p^{(K)}_{\min}q,\left(p^{(K)}_{\min}\right)\left(p^{(K)}_{\min}-3R^{(K)}\right)\right\}\right)\right),\\
&T=\Omega\left(\left(\Delta^{(K)}_{\lambda}\right)^{-1} \log\left(\left(p^{(K)}_{\min}\right)^{-1}\max\left\{\left(p^{(K)}_{\min}-3R^{(K)}\right)^{-1},\left(p^{(K)}_{\min}\right)^{-1}q^{-2}\right\}\right)\right),\\
 &|E|=\mathcal{O}\left(\min\left\{\left(p^{(K)}_{\min}\right)^2q^2,p^{(K)}_{\min}\left(p^{(K)}_{\min}-3R^{(K)}\right)\right\}\right)\,,\\
 \end{aligned}
\end{equation}
then, for each $m\in\mathcal{D}$, there must exist a unique $1\leq k_m\leq K$ such that 
\begin{equation}\label{eqn:close_gaussian_lemma}
     \left|\lambda_{m}-\theta^*_{k_m}\right|\leq \frac{q}{T},\quad \left|p_m-r^*_{k_m}\right|\leq q p_m\,.
\end{equation}
\end{lem}
Lemma \ref{lem:single_op_problem} constitutes a key element in the remaining part of the proof. According to Lemma \ref{lem:single_op_problem}, when the expectation error is small enough, the error of refined dominant eigenvalue estimation is bounded by $q/T$, where $q$ is a fixed (maybe small) constant and $T$ is the maximal runtime. Combining this lemma with Lemma \ref{lem:expectation_error} by setting $q=1$, we can  demonstrate that when $T_0=\Omega\left(\left(\Delta^{(K)}_{\lambda}\right)^{-1}\right)$ and $N_0=\Omega(T^2_0)$, solving the optimization problem gives us a reasonable approximation to the refined dominant eigenvalues, meaning: for each $m\in\mathcal{D}$, there must exist a unique $1\leq k_m\leq K$ such that $\left|\lambda_{m}-\theta^*_{k_m}\right|=\Or\left(\Delta^{(K)}_{\lambda}\right)$.

\begin{proof}[Proof of Lemma \ref{lem:single_op_problem}] First, we define \[F^*(x)=\exp\left(-\frac{T^2x^2}{2}\right)=\int^\infty_{-\infty}\frac{1}{\sqrt{2\pi}T}\exp\left(-\frac{x^2}{2T^2}\right)\exp(ixt)\rd t\,.\] Notice
\[
\begin{aligned}
L_K\left(\vr,\vtheta\right)=&\int^\infty_{-\infty}\frac{1}{\sqrt{2\pi}T}\exp\left(-\frac{t^2}{2T^2}\right)\left|\sum^{M}_{m=1}p_m\exp(-i \lambda_m t)-\sum^{K}_{k=1}r_k\exp(-i\theta_k t)\right|^2\rd t\\
&+E_{p,r}+E_{r,r}+E_{r,Z}+E_F\,,
\end{aligned}
\]
where $E_F=\int^\infty_{-\infty}\left(a(t)-\frac{1}{\sqrt{2\pi}T}\exp\left(-\frac{t^2}{2T^2}\right)\right)\left|\sum^{M}_{m=1}p_m\exp(-i \lambda_m t)-\sum^{K}_{k=1}r_k\exp(-i\theta_k t)\right|^2\rd t$.
Using the tail bound of a Gaussian and the choice of $\gamma$ \eqref{eqn:condition_gaussian_lemma}, we have
\begin{equation}\label{eqn:bound_E_F}
|E_F|=\mathcal{O}\left(\exp(-\gamma^2/2)\right)=\mathcal{O}\left(\min\left\{\left(p^{(K)}_{\min}q\right)^2,\left(p^{(K)}_{\min}-3R^{(K)}\right) p^{(K)}_{\min}\right\}\right)\,.
\end{equation}
Notice 
\begin{equation}\label{eqn:upper_bound_L_gaussian}
\begin{aligned}
&L_K\left(\left\{p_m\right\}_{m\in\mathcal{D}},\left\{\lambda_m\right\}_{m\in\mathcal{D}}\right)\\
\leq &\int^\infty_{-\infty}a(t)\left|\sum_{m\in\mathcal{D}^c}p_m\exp(-i\lambda_mt)\right|^2\rd t+\hat{E}+\hat{E}_F\\
\leq &\int^\infty_{-\infty}a(t)\left|\sum^{M}_{m=1}p_m\exp(-i\lambda_mt)-\sum_{m\in\mathcal{D}}p_m\exp(-i\lambda_mt)\right|^2\rd t+\hat{E}+\hat{E}_F\\
\leq &(R^{(K)})^2+\hat{E}+\hat{E}_F\,,
\end{aligned}
\end{equation}
where we $\hat{E}=E\left(\left\{p_m\right\}_{m\in\mathcal{D}},\left\{\lambda_m\right\}_{m\in\mathcal{D}}\right)$, and $\hat{E}_F=E_F\left(\left\{p_m\right\}_{m\in\mathcal{D}},\left\{\lambda_m\right\}_{m\in\mathcal{D}}\right)$.

We separate the proof into two steps. In the first step, we show that for each $m\in\mathcal{D}$, there must exist a unique $1\leq k_m\leq K$ such that 
\begin{equation}\label{eqn:prior_close_gaussian}
    \left|\lambda_{m}-\theta^*_{k_m}\right|\leq \frac{\Delta^{(K)}_{\lambda}}{4}.
\end{equation}
In the second step, we further improve the bound in \eqref{eqn:prior_close_gaussian}.

\textbf{Step 1: Show each $\theta^*_k$ should be close to one $\lambda_m$ for $m\in\mathcal{D}$.}

Suppose there exists $m^*\in\mathcal{D}$ such that for any $1\leq k\leq K$,
\[
\left|\lambda_{m^*}-\theta^*_k\right|> \frac{\Delta^{(K)}_{\lambda}}{4}\,.
\]
Then, let $E^*=E\left(\vr^*,\vtheta^*\right)$ and $E^*_F=E_F\left(\vr^*,\vtheta^*\right)$, 
\[
\begin{aligned}
&L_K\left(\vr^*,\vtheta^*\right)\\
= &\int^\infty_{-\infty}a(t)\left|\left(\left(\sum_{m\in\mathcal{D},m\neq m^*}p_{m}\exp(-i \lambda_{m} t)-\sum^{K}_{k=1}r^*_k\exp(-i\theta^*_k t)\right)+\sum_{m\in\mathcal{D}^c}p_m\exp(-i \lambda_m t)\right)\right.\\
&\left.+p_{m^*}\exp(-i \lambda_{m^*} t)\right|^2\rd t+E^*+E^*_F\\
\geq &p^2_{m^*}-2p_{m^*}R^{K
}-16\exp\left(-\frac{\left(T\Delta^{(K)}_{\lambda}\right)^2}{16}\right)+E^*+E^*_F\\
>& \left(R^{(K)}\right)^2+\hat{E}+\hat{E}_F\\
\geq &L_K\left(\left\{p_m\right\}_{m\in\mathcal{D}},\left\{\lambda_m\right\}_{m\in\mathcal{D}}\right).
\end{aligned}
\]
In the second inequality, we use $p_{m^*}>3R^{(K)}$,  \eqref{eqn:condition_gaussian_lemma}, and  \eqref{eqn:bound_E_F}. In the first inequality, we use the decaying property of $F(x)$ to obtain
\[
\begin{aligned}
&\left|\int^\infty_{-\infty}a(t)\left(\sum_{m\in\mathcal{D},m\neq m^*}p_{m}\exp(-i \lambda_{m} t)-\sum^{K}_{k=1}r^*_k\exp(-i\theta^*_k t)\right)p_{m^*}\exp(-i \lambda_{m^*} t)\rd t\right|\leq 8\exp\left(-\frac{\left(T\Delta^{(K)}_{\lambda}\right)^2}{16}\right)\,,\\
&\left|\int^\infty_{-\infty}a(t)\left(\sum_{m\in\mathcal{D}^c}p_m\exp(-i \lambda_m t)\right)p_{m^*}\exp(-i \lambda_{m^*} t)\rd t\right|\leq p_{m^*}R^{(K)}
\end{aligned}
\]
where we use $|\lambda_{m}-\lambda_{m^*}|>\Delta^{(K)}_{\lambda}/2$, $|\theta^*_{k}-\lambda_{m^*}|>\Delta^{(K)}_{\lambda}/4$.
This contradicts the assumption that $\left(\{r^*_k\}^{K}_{k=1},\{\theta^*_k\}^{K}_{k=1}\right)$ is the minimization point. Thus, \eqref{eqn:prior_close_gaussian} is true for all $1\leq k\leq K$.

\textbf{Step 2: Improve the upper bound.}

Define 
\[
\mathcal{F}^*=\int^\infty_{-\infty}\frac{1}{\sqrt{2\pi}T}\exp\left(-\frac{x^2}{2T^2}\right)\left|\sum_{m\in\mathcal{D}} p_{m}\exp\left(-i \lambda_{m} t\right)-r^*_{k_m}\exp\left(-i \theta^*_{k_m} t\right)\right|^2\rd t\,.
\]
We have
\[
\begin{aligned}
&L_K\left(\vr^*,\vtheta^*\right)\\
=&\int^\infty_{-\infty}a(t)\left|\left(\sum_{m\in\mathcal{D}}p_{m}\exp(-i \lambda_{m} t)-r^*_{k_m}\exp(-i\theta^*_{k_m} t)\right)+\sum_{m\in\mathcal{D}^c}p_m\exp(-i \lambda_m t)\right|^2+E^*+E^*_F\\
\geq &\mathcal{F}^*-2(\mathcal{F}^*)^{1/2}R^{(K)}+\int^\infty_{-\infty}a(t)\left|\sum_{m\in\mathcal{D}^c}p_m\exp(-i\lambda_mt)\right|^2\rd t+E^*+E^*_F\\
\end{aligned}
\]
Noticing that $\left(\vr^*,\vtheta^*\right)$ is the minimum point and comparing the above inequality with the second inequality of \eqref{eqn:upper_bound_L_gaussian}, we have 
\begin{equation}\label{eqn:upperbound_F_star}
\mathcal{F}^*-2(\mathcal{F}^*)^{1/2}R^{(K)}\leq \left|\hat{E}-E^*\right|+\left|\hat{E}_F-E^*_F\right|
\end{equation}
Using \eqref{eqn:condition_gaussian_lemma} and conditions of Theorem \ref{thm:gaussian}, we have $\left|\hat{E}\right|,\left|E^*\right|,\left|\hat{E}_F\right|,\left|E^*_F\right|=\Or\left(\left(p^{(K)}_{\min}q\right)^2\right)$. This implies
\[
\mathcal{F}^*-2(\mathcal{F}^*)^{1/2}R^{(K)}=\mathcal{O}\left(\left(p^{(K)}_{\min}q\right)^2\right)\,,
\]
Since $q=\Omega(R^{(K)}/p^{(K)}_{\min})$, we further have $\mathcal{F}^*=\mathcal{O}\left(\left(p^{(K)}_{\min}q\right)^2\right)$. 

Using \eqref{eqn:prior_close_gaussian}, $\sum_{k} |r^*_k|\leq 1$, and $\exp\left(-(\Delta^{(K)}_{\lambda} T)^2/8\right)=\mathcal{O}\left(\left(p^{(K)}_{\min}q\right)^2\right)$, we can show
\begin{equation}\label{eqn:interaction_bound}
\begin{aligned}
&\begin{aligned}
\sum_{m,m'\in\mathcal{D},m\neq m'}\left|\int^\infty_{-\infty}\frac{1}{\sqrt{2\pi}T}\exp\left(-\frac{x^2}{2T^2}\right)\right.&\left(p_{m}\exp\left(-i \lambda_{m} t\right)-r^*_{k_m}\exp\left(-i \theta^*_{k_m} t\right)\right)\\
\cdot&\left.\left(p_{m'}\exp\left(-i \lambda_{m'} t\right)-r^*_{k_{m'}}\exp\left(-i \theta^*_{k_{m'}} t\right)\right)\rd t\right|\\
\end{aligned}\\
\leq &\sum_{m,m'\in\mathcal{D},m\neq m'} p_mp_{m'}\exp\left(-\frac{T^2(\lambda_m-\lambda_{m'})^2}{2}\right)+p_m r^*_{k_{m'}}\exp\left(-\frac{T^2(\lambda_m-\theta^*_{k_{m'}})^2}{2}\right)\\
&+r^*_{k_m} p_{m'}\exp\left(-\frac{T^2(\theta^*_{k_{m}}-\lambda_{m'})^2}{2}\right)+r^*_{k_m} r^*_{k_{m'}}\exp\left(-\frac{T^2(\lambda_m-\theta^*_{k_{m'}})^2}{2}\right)\\
\leq &4\exp\left(-\frac{(\Delta^{(K)}_{\lambda} T)^2}{8}\right)=\mathcal{O}\left(\left(p^{(K)}_{\min}q\right)^2\right)\,.
\end{aligned}
\end{equation}
Expanding $\mathcal{F}^*$, we find that
\begin{equation}\label{eqn:F_star_expansion}
\begin{aligned}
    \mathcal{F}^*\geq &\sum_{m\in\mathcal{D}}\int^\infty_{-\infty}\frac{1}{\sqrt{2\pi}T}\exp\left(-\frac{t^2}{2T^2}\right)\left|p_{m}\exp\left(-i \lambda_{m} t\right)-r^*_{k_m}\exp\left(-i \theta^*_{k_m} t\right)\right|^2\rd t\\
    &-2\sum_{m,m'\in\mathcal{D},m\neq m'}\left|\int^\infty_{-\infty}\frac{1}{\sqrt{2\pi}T}\exp\left(-\frac{x^2}{2T^2}\right)\left(p_{m}\exp\left(-i \lambda_{m} t\right)-r^*_{k_m}\exp\left(-i \theta^*_{k_m} t\right)\right)\right.\\
\cdot&\left.\left(p_{m'}\exp\left(-i \lambda_{m'} t\right)-r^*_{k_{m'}}\exp\left(-i \theta^*_{k_{m'}} t\right)\right)\rd t\right|\,.
\end{aligned}
\end{equation}
By utilizing the bound $\mathcal{F}^*=\mathcal{O}\left(\left(p^{(K)}_{\min}q\right)^2\right)$ and \eqref{eqn:interaction_bound}, we can infer from \eqref{eqn:F_star_expansion} that
\begin{equation}\label{eqn:key_F_gaussian}
\begin{aligned}
\sum_{m\in\mathcal{D}}\int^\infty_{-\infty}\frac{1}{\sqrt{2\pi}T}\exp\left(-\frac{t^2}{2T^2}\right)\left|p_{m}\exp\left(-i \lambda_{m} t\right)-r^*_{k_m}\exp\left(-i \theta^*_{k_m} t\right)\right|^2\rd t=\mathcal{O}\left(\left(p^{(K)}_{\min}q\right)^2\right)\,.
\end{aligned}
\end{equation}
This implies
\[
\begin{aligned}
&\sum_{m\in\mathcal{D}}p^2_{m}\left(1-\exp\left(-T^2\left(\theta^*_{k_m}-\lambda_{m}\right)^2\right)\right)\\
\leq &\sum_{m\in\mathcal{D}}\int^\infty_{-\infty}\frac{1}{\sqrt{2\pi}T}\exp\left(-\frac{t^2}{2T^2}\right)\left|p_{m}\exp\left(i \left(\theta^*_{k_m}-\lambda_{m}\right) t\right)-r^*_{k_m}\right|^2\rd t\\
= &\sum_{m\in\mathcal{D}}\int^\infty_{-\infty}\frac{1}{\sqrt{2\pi}T}\exp\left(-\frac{t^2}{2T^2}\right)\left|p_{m}\exp\left(-i \lambda_{m} t\right)-r^*_{k_m}\exp\left(-i \theta^*_{k_m} t\right)\right|^2\rd t\\
=  &\mathcal{O}\left(\left(p^{(K)}_{\min}q\right)^2\right)\,.
\end{aligned}
\]
Because $p_m\geq p^{(K)}_{\min}$ for $m\in\mathcal{D}$, we obtain that, for each $k$, $1-\exp\left(-T^2\left(\lambda_{m}-\theta^*_{k_m}\right)^2\right)=\mathcal{O}\left(q^2\right)$, which implies the first inequality of \eqref{eqn:close_gaussian_lemma}. For the second inequality of \eqref{eqn:close_gaussian_lemma}, \eqref{eqn:key_F_gaussian} also implies
\[
\left|p_{m}\exp\left(-T^2\left(\theta^*_{k_m}-\lambda_{m}\right)^2/2\right)-r^*_{k_m}\right|^2=\mathcal{O}\left(\left(p^{(K)}_{\min}q\right)^2\right)
\]
Because $\left|\exp\left(-T^2\left(\theta^*_{k_m}-\lambda_{m}\right)^2/2\right)-1\right|=\mathcal{O}\left(q^2\right)$, we have $|p_{m}-r^*_{k_m}|\leq q p_{m}$. This concludes the proof of the second inequality of \eqref{eqn:close_gaussian_lemma}.
\end{proof}

The following proposition controls the complexity of the second step of the algorithm:
\begin{prop}\label{prop:single_op_problem}
Given  failure probability $0<\eta<1/2$, any small constant $\zeta>0$,  $0<q<1$ such that $q=\Theta(R^{(K)}/p^{(K)}_{\min})$, and given a sequence of rough intervals $\{I_k\}^{K}_{k=1}\subset \mathbb{R}$, we assume 1. $|I_k|\leq 4\pi/T$; 2. for any $m\in\mathcal{D}$, there exists a unique $1\leq k_m\leq K$ such that $\lambda_m\in I_{k_m}$; 3. $p^{(K)}_{\min}>3R^{(K)}$. Define
  \[
  \left(\vr^*,\vtheta^*\right)=\argmin_{\|\vr\|_1\leq 1,\theta_k\in I_k}L_{K}\left(\vr,\vtheta\right)\,,
  \]
where $L_{K}$ is defined in \eqref{eqn:loss_multi_modal}. If
\begin{equation}\label{eqn:condition_gaussian_prop}
\begin{aligned}
&\gamma =\Theta\left(\log\left(\left(p^{(K)}_{\min}\right)^{-1}q^{-1}\right)\right)\,,\\
&T=\Omega\left(\left(\Delta^{(K)}_{\lambda}\right)^{-1}\log\left(\left(p^{(K)}_{\min}\right)^{-1}q^{-1}\right)\right),\\
&N=\Theta\left(\gamma^2\left(p^{(K)}_{\min}\right)^{-4}q^{-(2+\zeta)}\mathrm{polylog}\left(K\log(\zeta^{-1})\left(p^{(K)}_{\min}\right)^{-1}\eta^{-1}\right)\right)\,,
\end{aligned}
\end{equation}
then with probability $1-\eta$, 
\begin{equation}\label{eqn:final_close_gaussian_prop}
    |\lambda_m-\theta^*_{k_m}|\leq \frac{q}{T},\quad \quad |p_{m}-r^*_{k_m}|\leq q p_m\,.
\end{equation}
\end{prop}
We give the proof of Proposition \ref{prop:single_op_problem} in Appendix \ref{sec:pf_prop}. Here, we emphasize that Proposition \ref{prop:single_op_problem} can not be directly proved by combining Lemma \ref{lem:single_op_problem} and Lemma \ref{lem:expectation_error}. According to Lemma \ref{lem:single_op_problem}, to obtain the accuracy $q/T$, we need the expectation error $|E|=\Or(q^2)$. According to Lemma \ref{lem:expectation_error}, to guarantee $|E|=\Or(q^2)$, we need to choose $N=\Omega(q^{-4})$, which is worse than the scaling of $N$ with respect to $q^{-1}$ in the proposition. This means that to achieve the correct scaling of $N$ with respect to $q^{-1}$, we need to bound the expectation error in a different way to obtain the sharp estimation.

Now, we are ready to use Lemma \ref{lem:single_op_problem} and Proposition \ref{prop:single_op_problem} to prove Theorems \ref{thm:gaussian}.
\begin{proof}[Proof of Theorem \ref{thm:gaussian}] First, according to the definition of $T_0$, we have $T_0=\widetilde{\Theta}((\Delta^{\lambda}_K)^{-1}\log(q^{-1}))$. Combining \eqref{eqn:condition_gaussian} with Lemma \ref{lem:single_op_problem} and Lemma \ref{lem:expectation_error} (by setting $q=1$ in Lemma \ref{lem:single_op_problem} and $\rho=\pi T_0$ and $\xi=\Theta\left(\left(p^{(K)}_{\min}-3R^{(K)}\right) p^{(K)}_{\min}\right)$  in Lemma \ref{lem:expectation_error} \eqref{eqn:exp_expectation_error}), after step 1, with probability $1-\eta/(l+1)$, we obtain that, for every $m\in\mathcal{D}$, there exists an unique $1\leq k_m\leq K$ such that 
    \begin{equation}\label{eqn:first_close_sin_2_refine}
        |\lambda_{m}-\theta^*_{k_m}|\leq \frac{1}{T_0}<\min\left\{\frac{\Delta^{(K)}_\lambda}{4},\frac{\pi}{T_0}\right\},\quad |p_{m}-r^*_{k_m}|\leq p_{m}\,.
    \end{equation}
Thus, with probability $1-\eta/(l+1)$, after step 1, $I_{k_m}=[\lambda_{\min,k_m},\lambda_{\max,k_m}]$ satisfies the condition of Proposition \ref{prop:single_op_problem} with $T=T_1$.

Next, for $j=1$ in step 2, using Proposition \ref{prop:single_op_problem} by setting $T=T_1$, we obtain that, with probability $1-\eta/(l+1)$,
\begin{equation}\label{eqn:second_close_sin_2_refine}
        |\lambda_{m}-\theta^*_{k_m}|\leq \frac{q}{T_1}<  \frac{\pi}{T_1},\quad |p_{m}-r^*_{k_m}|\leq q p_{m}\,.
\end{equation}
This implies that, after step 2 with $j=1$, with probability $1-2\eta/(l+1)$, $I_{k_m}=[\lambda_{\min,k_m},\lambda_{\max,k_m}]$ satisfies the condition of Proposition \ref{prop:single_op_problem} with $T=T_2$. Using this recursively, we finally obtain that, when $j=l$, with probability $1-\eta$, we have 
\[
        |\lambda_{m}-\theta^*_{k_m}|\leq \frac{q}{T_l}<\epsilon,\quad |p_{m}-r^*_{k_m}|\leq q p_{m}\,.
\]
which implies \eqref{eqn:final_close_gaussian}.

Finally, using the choices of $\gamma,l,T_l$ and $N_l$, we obtain
\[
T_{\max}= T_l\gamma=\Theta\left(\frac{q}{\epsilon}\log\left(\frac{1}{\min\left\{p^{(K)}_{\min}q,\left(p^{(K)}_{\min}\right)\left(p^{(K)}_{\min}-3R^{(K)}\right)\right\}}\right)\right)=\frac{\delta}{\epsilon}
\]
and
\[T_{\mathrm{total}}=\widetilde{\Theta}\left(\frac{1}{\left(p^{(K)}_{\min}\right)^{4}q^{1+\zeta}\epsilon}\polylog(\log(\zeta^{-1})\eta^{-1})\right)=\wt{\Theta}\left(\frac{\polylog(\log(\zeta^{-1})\delta^{-1}\eta^{-1})}{\left(p^{(K)}_{\min}\right)^4\delta^{1+\zeta}\epsilon}\right)\,,
\]
where $\delta=\widetilde{\Theta}\left(q\log(q^{-1})\right)$. 
\end{proof}

\section{Proof of Proposition \ref{prop:single_op_problem}}\label{sec:pf_prop}

In this section, we prove Proposition \ref{prop:single_op_problem}. The proof idea is similar to \cite[Appendix B.2]{DingLin2023}. We first give a rough complexity result that has a $\Or(q^{-4})$ scaling in $N$. Then, we consider the difference of the expectation error more carefully and improve this scaling to $\Or(q^{-2-o(1)})$ using an iteration argument.

Applying \cref{eqn:exp_expectation_error} from  Lemma \ref{lem:expectation_error} in Appendix \ref{sec:bound_expectation_error} to bound the expectation error $E(\vr,\vtheta)$, we obtain the following lemma: 
\begin{lem}\label{lem:rough_bound} Given  failure probability $0<\eta<1/2$, any small constant $\zeta>0$,  $0<q<1$ such that $q=\Omega(R^{(K)}/p^{(K)}_{\min})$, and given a sequence of rough interval $\{I_k\}^{K}_{k=1}\subset \mathbb{R}$, we assume 1. $|I_k|\leq 4\pi/T$; 2. for any $m\in\mathcal{D}$, there exists a unique $1\leq k_m\leq K$ such that $\lambda_m\in I_{k_m}$; 3. $p^{(K)}_{\min}>3R^{(K)}$. Define
  \[
  \left(\vr^*,\vtheta^*\right)=\argmin_{\|\vr\|\leq 1,\theta_k\in I_k}L_{K}\left(\vr,\vtheta\right)\,,
  \]
where $L_{K}$ is defined in \eqref{eqn:loss_multi_modal}. If
    \begin{equation}\label{eqn:condition_lemma_2_gaussian}
    \begin{aligned}
    &\gamma =\Omega\left(\log\left(\left(p^{(K)}_{\min}\right)^{-1}q^{-1}\right)\right)\,,\\
    &T=\Omega\left(\left(\Delta^{(K)}_{\lambda}\right)^{-1}\log\left(\left(p^{(K)}_{\min}\right)^{-1}q^{-1}\right)\right)\,,\\
   &N=\Omega\left(\gamma^2\left(p^{(K)}_{\min}\right)^{-4}q^{-4}\mathrm{polylog}\left(K\left(p^{(K)}_{\min}\right)^{-1}q^{-1}\eta^{-1}\right)\right)\,,
    \end{aligned}
    \end{equation}
    then with probability $1-\eta$, 
    \begin{equation}\label{eqn:final_lemma_2_gaussian}
        \left|\lambda_m-\theta^*_{k_m}\right|\leq \frac{q}{T},\quad \quad |p_{m}-r^*_{k_m}|\leq q p_m\,.
    \end{equation}
\end{lem}
\begin{proof}[Proof of Lemma \ref{lem:rough_bound}] According to the second equality of \eqref{eqn:condition_lemma_2_gaussian}, we have
\[
|I_k|\leq \frac{4\pi}{T}\leq \frac{\Delta^{(K)}_{\lambda}}{4}\,,
\]
which implies that for all $m\in\mathcal{D}$
\begin{equation}\label{eqn:rough_bound}
    \left|\lambda_{m}-\theta^*_{k_m}\right|\leq \frac{\Delta^{(K)}_{\lambda}}{4}\,.
\end{equation}
Using \eqref{eqn:condition_lemma_2_gaussian} and Lemma \ref{lem:expectation_error} \eqref{eqn:exp_expectation_error} by setting $\xi=\mathcal{O}\left(\left( p^{(K)}_{\min}q\right)^2\right)$, we can conclude that with probability $1-\eta$, $|E|=\mathcal{O}\left(\left( p^{(K)}_{\min}q\right)^2\right)$. Finally, since we have obtained a rough bound \eqref{eqn:rough_bound}, we can use the same argument as the second step in the proof of Lemma \ref{lem:single_op_problem} to prove \eqref{eqn:final_lemma_2_gaussian}.
\end{proof}

Next, in order to improve the scaling $\mathcal{O}(q^{-4})$ in \eqref{eqn:condition_lemma_2_gaussian}. We propose a different approach to bound the error terms. Instead of bounding $E(\vr^*,\vtheta^*)$ and $E(\{p_m\},\{\lambda_m\})$ separately,  we aim to bound the difference between these two error terms. Intuitively, when $(r^*_{k_m},\theta^*_{k_m})$ and $(p_m,\lambda_m)$ are close to each other, the two error terms are likely to cancel each other out when we compare the difference between $L(\vr^*,\vtheta^*)$ and $L(\{p_m\},\{\lambda_m\})$. This intuition is supported by Lemma \ref{lem:expectation_error} \eqref{eqn:exp_expectation_difference}.  Assuming that we already have $\left|\theta^*_{k_m}-\lambda_m\right|<\frac{q}{T}$ and $|r_{k_m}-p_m|\leq q p_m$, then it is sufficient to choose $N=\widetilde{\Omega}\left(q^2\xi^{-2}\right)$  to ensure that  $\left|E(\vr^*,\vtheta^*)-E( \{p_m\},\{\lambda_m\})\right|\geq \xi$ with high probability. This requirement, compared with the second inequality of \eqref{eqn:condition_lemma_2_gaussian}, reduces the blow-up rate to $\Or(q^{-2})$ as $q\rightarrow0$, which matches the condition in \cref{prop:single_op_problem}. However, the above calculation is assuming $\left|\theta^*_{k_m}- \lambda_m\right|<\frac{q}{T}$ and $|r_{k_m}- p_m|\leq q p_m$, which is unknown to us in advance. To overcome this difficulty, we need to use an iteration argument to obtain the desired order. This is summarized in the following lemma:

\begin{lem}\label{lem:improve_dependence_2}

Given  failure probability $0<\eta<1/2$, an integer $S>1$, any small constant $\zeta>0$, a sequence of rough interval $\{I_k\}^{K}_{k=1}\subset \mathbb{R}$, and a decreasing sequence $\{q_s\}^S_{s=0}$ with $0<q_0\leq 1$ and $q_S=\Omega(R^{(K)}/p^{(K)}_{\min})$, we assume 1. $|I_k|\leq 4\pi/T$; 2. for any $m\in\mathcal{D}$, there exists a unique $1\leq k_m\leq K$ such that $\lambda_m\in I_{k_m}$; 3. $p^{(K)}_{\min}>3R^{(K)}$; 4. \eqref{eqn:condition_lemma_2_gaussian} holds with $q=q_0$. If
\begin{equation}\label{eqn:condition_lemma_4_gaussian}
\begin{aligned}
&\gamma =\Omega\left(\log\left(\left(p^{(K)}_{\min}\right)^{-1}q^{-1}_S\right)\right)\,,\\
&T=\Omega\left(\left(\Delta^{(K)}_{\lambda}\right)^{-1}\log\left(\left(p^{(K)}_{\min}\right)^{-1}q^{-1}_S\right)\right)\,,\\
&N=\Omega\left(\max_{0\leq s\leq S-1}\left\{\left(p^{(K)}_{\min}\right)^{-4}q^{-4}_{s+1}q^2_s\mathrm{polylog}\left(KQ\left(p^{(K)}_{\min}\right)^{-1}q^{-1}_{s+1}q_s\eta^{-1}\right)\right\}\right)\,,
\end{aligned}
\end{equation}
then with probability $1-\eta$, 
\begin{equation}\label{eqn:final_lemma_4_gaussian}
    \left|\lambda_m-\theta^*_{k_m}\right|\leq \frac{q_S}{T},\quad |p_{m}-r^*_{k_m}|\leq q_S p_m\,.
\end{equation}
\end{lem}

\begin{proof}[Proof of Lemma \ref{lem:improve_dependence_2}] By utilizing Lemma \ref{lem:rough_bound} in conjunction with \eqref{eqn:condition_lemma_2_gaussian}, we can demonstrate that with probability $1-\eta/(2Q)$,
\[
 \left|\lambda_{m}-\theta^*_{k_m}\right|\leq \frac{q_0}{T},\quad \left|p_{m}-r^*_{k_m}\right|\leq q_0 p_m\,.
\]
Using these inequalities with Lemma \ref{lem:expectation_error} \eqref{eqn:exp_expectation_difference}, where we have $\rho=q_0$, $\xi=\left(p^{(K)}_{\min}q_1\right)^2 q^{-1}_0$, we obtain that with probability $1-\eta/Q$, $\left|\hat{E}-E^*\right|,\left|\hat{E}_F\right|,\left|E^*_F\right|=\Or\left(\left(p^{(K)}_{\min}q_1\right)^2\right)$. Then, similar to the second step in the proof of Lemma \ref{lem:single_op_problem} (plugging the bound of $\left|\hat{E}-E^*\right|,\left|\hat{E}_F\right|,\left|E^*_F\right|$ into \eqref{eqn:upperbound_F_star}), we can show that with probability $1-\eta/Q$,
\[
 \left|\lambda_{m}-\theta^*_{k_m}\right|\leq \frac{q_1}{T},\quad \left|p_{m}-r^*_{k_m}\right|\leq q_1 p_m\,.
\]
Next, similar to previous argument, using these inequalities with Lemma \ref{lem:expectation_error} \eqref{eqn:exp_expectation_difference}, where we have $\rho=q_1,\xi=\left(p^{(K)}_{\min}q_2\right)^2 q^{-1}_1$, we obtain that with probability $1-3\eta/(2Q)$, $\left|\hat{E}-E^*\right|,\left|\hat{E}_F\right|,\left|E^*_F\right|=\Or\left(\left(p^{(K)}_{\min}q_2\right)^2\right)$. This further implies, with probability $1-2\eta/Q$,
\[
 \left|\lambda_{m}-\theta^*_{k_m}\right|\leq \frac{q_2}{T},\quad \left|p_{m}-r^*_{k_m}\right|\leq q_2 p_m\,.
\]
Doing this repeatedly, we finally obtain \eqref{eqn:final_lemma_4_gaussian}.
\end{proof}

Finally, we are ready to prove the proposition \ref{prop:single_op_problem}:
\begin{proof}[Proof of Proposition \ref{prop:single_op_problem}]
Using the given $q$ from the conditions of the proposition, we construct a decreasing positive sequence $\{q_s\}^S_{s=0}$ with $q_s=q^{\frac{2-(1/2)^s}{2-(1/2)^S}}$. With this choice, we have $q^{-4}_{s+1}q^2_s=q^{-\frac{4}{2-(1/2)^S}}$. 

Setting $\zeta=(1/2)^{S-1}>\frac{4}{2-(1/2)^S}-2$ and using Lemma \ref{lem:improve_dependence_2} \eqref{eqn:final_lemma_4_gaussian}, we prove \eqref{eqn:final_close_gaussian_prop}.
\end{proof}

\section{Bound of the expectation error}\label{sec:bound_expectation_error}
In this section, we bound the expectation error. Recall the definition in \eqref{eqn:formula_E},
\[
E_{p,r}\left(\vr,\vtheta\right)=2\mathrm{Re}\left(\sum^{M}_{m=1}\sum^{K}_{k=1}p_mr_k\left(\frac{1}{N}\sum^N_{n=1}\exp(i(\theta_k- \lambda_m) t_n)-F(\theta_k-\lambda_m)\right)\right)\,,
\]
\[
E_{r,r}\left(\vr,\vtheta\right)=2\mathrm{Re}\left(\sum^{K}_{k_1\neq k_2}\overline{r}_{k_1}r_{k_2}\left(\frac{1}{N}\sum^N_{n=1}\exp(i(\theta_{k_2}-\theta_{k_1}) t_n)-F(\theta_{k_2}-\theta_{k_1})\right)\right)\,,
\]
\[
E_{r,Z}\left(\vr,\vtheta\right)=-2\mathrm{Re}\left(\sum^{K}_{k=1}\overline{r_k}\left(\frac{1}{N}\sum^N_{n=1}E_n\exp(i\theta_k t_n)\right)\right)\,.
\]
Intuitively, when $N\gg 1$, we have $|E_{p,r}|,|E_{r,r}|,|E_{r,Z}|=\mathcal{O}(1/\sqrt{N})$ for fixed $\vr,\vtheta$. However, this result does not directly apply to the optimization problem because the range of $(\vr,\vtheta)$ is always an infinite set. To overcome this difficulty, we notice that the Lipschitz constant of the function $
\exp(i\theta t)$ is bounded by $T$ if $t\leq T$. First, we use Hoeffding's inequality to obtain a uniform bound for these expectation errors with a finite number of $(\vr,\vtheta)$ points. Then, we extend this bound to all other points using the Lipschitz continuity property of $\exp(i\theta t)$. Specifically, we have the following result that gives a uniform bound for these expectation errors: 
\begin{lem}\label{lem:expectation_error}
Define $E_{p,r},E_{r,r},E_{r,Z}$ as above. Assume $a(t)$ is defined as \eqref{eqn:a_T}, then
    \begin{itemize}
        \item Given $0<\eta<1/2$ and a sequence of interval $\{I_k\}^{K}_{k=1}$ on $\mathbb{R}$, we define $\rho=\max_k\{T|I_k|/2\}$, if
        \[
        N=\Omega\left(\max\{\gamma^2\rho^2,1\}\xi^{-2}\mathrm{polylog}(K\xi^{-1}\eta^{-1})\right)\,,
        \]
        then
        \begin{equation}\label{eqn:exp_expectation_error}
        \begin{aligned}
        &\mathbb{P}\left(\sup_{\|\vr\|_1\leq 1,\theta_k\in I_k}|E_{p,r}\left(\vr,\vtheta\right)|\geq \xi\right)\leq \eta\\
        &\mathbb{P}\left(\sup_{\|\vr\|_1\leq 1,\theta_k\in I_k}|E_{r,r}\left(\vr,\vtheta\right)|\geq \xi\right)\leq \eta\\
        &\mathbb{P}\left(\sup_{\|\vr\|_1\leq 1,\theta_k\in I_k}|E_{Z}\left(\vr,\vtheta\right)|\geq \xi\right)\leq \eta
        \end{aligned}\,.
        \end{equation}
        \item Denote $\mathcal{D}=\{m_1,\dots,m_K\}$. Given $0<\eta<1/2$, a sequence of intervals $\{I_k\}^{K}_{k=1}$ on $\mathbb{R}$, and a sequence of discs $\{R_k\}^{K}_{k=1}$ on $\mathbb{C}$, we assume: 1. $ \lambda_{m_k}\in I_k$; 2. $p_{m_k}\in R_k$. Define $\rho=\max_k\{T|I_k|/2,\mathrm{radius}(R_k)/p_{m_k}\}$. If
        \[
        N=\Omega\left(\max\{\gamma^2\rho^2,1\}\xi^{-2}\mathrm{polylog}(K\xi^{-1}\eta^{-1})\right)\,,
        \]
        then
        \begin{equation}\label{eqn:exp_expectation_difference}
        \begin{aligned}
        &\mathbb{P}\left(\sup_{r_k\in R_k,\theta_k\in I_k}|E_{p,r}\left(\vr,\vtheta\right)-E_{p,r}(\{(p_{m_k},\lambda_{m_k})\}^{K}_{k=1})|\geq \rho\xi\right)\leq \eta\\
        &\mathbb{P}\left(\sup_{r_k\in R_k,\theta_k\in I_k}|E_{r,r}\left(\vr,\vtheta\right)-E_{r,r}(\{(p_{m_k},\lambda_{m_k})\}^{K}_{k=1})|\geq \rho\xi\right)\leq \eta\\
        &\mathbb{P}\left(\sup_{r_k\in R_k,\theta_k\in I_k}|E_{Z}\left(\vr,\vtheta\right)-E_{Z}(\{(p_{m_k},\lambda_{m_k})\}^{K}_{k=1})|\geq \rho\xi\right)\leq \eta
        \end{aligned}
        \end{equation}
    \end{itemize}
\end{lem}
\begin{proof}[Proof of Lemma \ref{lem:expectation_error}] We start by proving \eqref{eqn:exp_expectation_error}. We will only prove the inequalities for $E_{p,r}$, but note that the other two can be shown similarly. Define
\[
G(\theta)=\frac{1}{N}\sum^N_{n=1}\left(\sum^M_{m=1}p_m\exp(i(\theta- \lambda_m) t_n)\right)-\left(\sum^M_{m=1}p_m F(\theta-\lambda_m)\right)\,.
\]
For the truncated filter, since $t_n\leq \gamma T$, we have
\[
|G(\theta_1)-G(\theta_2)|\leq  \gamma T|\theta_1-\theta_2|\,,
\]
which implies $G$ is a $\gamma T$-Lipschitz function. Because $|F(\theta-\lambda_m)|\leq 1$, using Hoeffding's inequality, for fixed $\theta$, we have
\[
\mathbb{P}\left(|G(\theta)|\geq \xi\right)\leq 4\exp\left(-\frac{N\xi^2}{32}\right)\,.
\]
Combining this with the fact that $G$ is a $\gamma T$-Lipschitz function, we have
\[
\mathbb{P}\left(\sup_{\theta\in I_k}|G(\theta)|\geq \xi+\gamma T\epsilon\right)\leq \frac{8\rho}{T\epsilon}\exp\left(-\frac{N\xi^2}{32}\right)\,.
\]
Choosing $\epsilon=\frac{\rho}{T\sqrt{N}}$, $\xi=\frac{4\sqrt{2}}{\sqrt{N}}\log^{1/2}(8\sqrt{N}K/\eta)$, 
\[
\mathbb{P}\left(\sup_{\theta\in I_k}|G(\theta)|\geq \left(4\sqrt{2}\log^{1/2}(8\sqrt{N}K/\eta)+\gamma\rho\right)\frac{1}{\sqrt{N}}\right)\leq \frac{\eta}{K}\,,
\]
which implies the first inequality of \eqref{eqn:exp_expectation_error} using $\sum |r_k|\leq 1$. 

We proceed to prove \eqref{eqn:exp_expectation_difference}. As before, we only show the proof of the first inequality in \eqref{eqn:exp_expectation_difference}. Fixed $1\leq k\leq K$, when $r_k\in R_k$ and $\theta_k\in I_k$, we have
\[
r_kG(\theta_k)-p_{m_k}G(\lambda_{m_k})=(r_k-p_{m_k})G(\theta_k)+p_{m_k}(G(\theta_k)-G( \lambda_{m_k}))\,.
\]
For the first term, we have $|r_k-p_{m_k}|\leq 2\rho p_{m_k}$ and
\[
\mathbb{P}\left(\sup_{\theta_k\in I_k}\left|G(\theta_k)\right|\geq \xi/2\right)\leq \frac{\eta}{2K}\,,
\]
according to the previous proof. This implies
\begin{equation}\label{eqn:different_1}
\mathbb{P}\left(\sup_{r_k\in R_k,\theta_k\in I_k}\sum^{K}_{k=1}|r_k-p_{m_k}|\left|G(\theta_k)\right|\geq \rho\xi/2\right)\leq \frac{\eta}{2}\,.
\end{equation}

For the second term, we first notice that
\[
\left|\sum^M_{m=1}p_m\exp(i(\theta_k- \lambda_m) t_n)-\sum^M_{m=1}p_m\exp(i(\lambda_{m_k}- \lambda_m) t_n)\right|\leq \gamma T|\theta_k-\lambda_{m_k}|\leq 2\rho
\]
and $G(\theta)$ is a $\gamma T$-Lipschitz function. Then, similar to before, we have
\[
\mathbb{P}\left(\sup_{\theta_k\in I_k}\sum^{K}_{k=1}\left|G(\theta_k)-G(\lambda_{m_k})\right|\geq \rho\xi/2\right)\leq \frac{\eta}{2K}\,.
\]
Since $\sum_k p_{m_k}\leq 1$, we have
\[
\mathbb{P}\left(\sup_{r_k\in R_k,\theta_k\in I_k}\sum^{K}_{k=1}\left|p_{k}(G(\theta_k)-G(\lambda_{m_k}))\right|\geq \rho\xi/2\right)\leq \frac{\eta}{2}\,.
\]
Combining this inequality with \eqref{eqn:different_1}, we prove the first inequality of \eqref{eqn:exp_expectation_difference}. 
\end{proof}

\end{document}